\documentclass[12pt]{article}

\usepackage{geometry}
\usepackage{lmodern,enumerate,rotating,tensor}
\usepackage{amsmath,amsfonts,amssymb,subfigure,mathtools,mathabx,blkarray,centernot}
\usepackage{tikz-cd}
\usetikzlibrary{cd}
\usetikzlibrary{shapes,shapes.geometric,arrows,fit,calc,positioning,automata}
\usepackage[normalem]{ulem}

\usepackage{stmaryrd}

\DeclareMathOperator{\Acc}{Acc}
\DeclareMathOperator{\CCa}{CC_A}
\DeclareMathOperator{\CCauo}{CC_A^{uo}}
\DeclareMathOperator{\CCao}{CC_A^{o}}

\DeclareMathOperator{\last}{last}

\DeclareMathOperator{\WEI}{WT}

\DeclareMathOperator{\Tr}{Tr}

\DeclareMathOperator{\init}{init}

\DeclareMathOperator{\ED}{ED}
\DeclareMathOperator{\EaD}{EaD}

\newcommand{\nsf}{\mathsf{n}}
\newcommand{\fsf}{\mathsf{f}}

\allowdisplaybreaks

\newcommand{\Z}{\mathbb{Z}}
\newcommand{\R}{\mathbb{R}}
\newcommand{\N}{\mathbb{N}}
\newcommand{\Q}{\mathbb{Q}}

\newcommand{\underN}{\underline{\mathbb{N}}}
\newcommand{\underQ}{\underline{\mathbb{Q}}}
\newcommand{\underR}{\underline{\mathbb{R}}}
\newcommand{\undernonnegQ}{\underline{\mathbb{Q}_{\ge0}}}

\newcommand{\Acal}{\mathcal{A}}
\newcommand{\dt}{\delta}

\newcommand{\ep}{\epsilon}

\newcommand{\Sig}{\Sigma}
\newcommand{\s}{\sigma}

\newcommand{\frakG}{\mathfrak{G}}

\newcommand{\frakM}{\mathfrak{M}}
\newcommand{\frakun}{\mathfrak{u_n}}
\newcommand{\frakt}{\mathfrak{t}}

\newcommand{\frakq}{\mathfrak{q}}
\newcommand{\frakD}{\mathfrak{D}}
\newcommand{\frakDp}{\mathfrak{D_p}}

\newcommand{\PSPACE}{\mathsf{PSPACE}}

\newcommand{\EXPTIME}{\mathsf{EXPTIME}}
\newcommand{\PTIME}{\mathsf{PTIME}}
\newcommand{\EXPSPACE}{\mathsf{EXPSPACE}}

\newcommand{\NP}{\mathsf{NP}}
\newcommand{\coNP}{\mathsf{coNP}}

\newcommand{\llb}{\llbracket}
\newcommand{\rrb}{\rrbracket}

\tikzset{
->, 
node distance=3cm, 
every state/.style={thick, fill=gray!10}, 
initial text=$ $, 
}

\usetikzlibrary{decorations.pathmorphing}
\newcommand\xrsquigarrow[1]{%
    \mathrel{%
        \begin{tikzpicture}[%
            baseline={(current bounding box.south)}
            ]
        \node[%
            ,inner sep=.44ex
            ,align=center
            ] (tmp) {$\scriptstyle #1$};
        \path[%
            ,draw,<-
            ,decorate,decoration={%
                ,zigzag
                ,amplitude=0.7pt
                ,segment length=1.2mm,pre length=3.5pt
                }
            ] 
        (tmp.south east) -- (tmp.south west);
        \end{tikzpicture}
        }
    }

\newtheorem{remark}{Remark}

\newtheorem{example}{Example}
\usepackage{times,amsmath,amssymb,graphicx,tikz,algorithm,algorithmic,url,hhline,booktabs,soul}
\usetikzlibrary{automata,arrows,positioning}
\usetikzlibrary{petri}

\makeatletter
\newcommand{\subalign}[1]{%
  \vcenter{%
    \Let@ \restore@math@cr \default@tag
    \baselineskip\fontdimen10 \scriptfont\tw@
    \advance\baselineskip\fontdimen12 \scriptfont\tw@
    \lineskip\thr@@\fontdimen8 \scriptfont\thr@@
    \lineskiplimit\lineskip
    \ialign{\hfil$\m@th\scriptstyle##$&$\m@th\scriptstyle{}##$\crcr
      #1\crcr
    }%
  }
}
\makeatother

\usepackage[colorlinks]{hyperref}

\newcommand{\red}{\color{red}}
\newcommand{\blue}{\color{blue}}
\definecolor{green}{rgb}{0.1,0.7,0.1}
\newcommand{\green}{\color{green}}

\newcommand{\cyan}{\color{cyan}}

\newtheorem{theorem}{Theorem}[section]
\newtheorem{lemma}[theorem]{Lemma}
\newtheorem{corollary}[theorem]{Corollary}
\newtheorem{definition}[theorem]{Definition}
\newtheorem{proposition}[theorem]{Proposition}

\newtheorem{problem}{Problem}

\newcounter{enumi_saved}
\newenvironment{myenumerate} {
    \begin{enumerate}[(i)]\setcounter{enumi}{\value{enumi_saved}}}
    {\setcounter{enumi_saved}{\value{enumi}}\end{enumerate}}

\title{Diagnosability of labeled \texorpdfstring{$\frakDp$}{}-automata}

\author{Kuize Zhang\\
{\small Department of Computer Science}\\
{\small University of Surrey, Guildford GU2 7XH, United Kingdom}\\
{\small kuize.zhang@surrey.ac.uk}
\and 
J\"{o}rg Raisch\\
{\small Control Systems Group}\\
{\small Technische Universit\"{a}t Berlin, D-10587 Berlin, Germany}\\
{\small raisch@control.tu-berlin.de}
}

\begin{document}

\date{}

\maketitle

{\bf Abstract}
    In this paper, we formulate a notion of diagnosability for labeled weighted automata over 
	a class of dioids which admit both positive and negative numbers as well as vectors. The weights can 
	represent diverse physical meanings such as time elapsing and position deviations. We also develop
	an original tool called concurrent composition to verify diagnosability for such automata.
	These results are fundamentally new compared with the existing ones in the literature.

    In a little more detail, \emph{diagnosability} is characterized for a labeled weighted automaton $\Acal^{\frakDp}$
	over a  special dioid $\frakDp$ called \emph{progressive}, which can represent diverse physical meanings such
	as time elapsing and position  deviations. In a progressive dioid, the canonical order is total,
	there is at least one \emph{eventually dominant} element, there is no zero divisor, and
	the cancellative law is satisfied, where the functionality of an eventually dominant element $t$ is
	to make every nonzero element $a$ arbitrarily large by multiplying $a$ by $t$ for sufficiently many times.
	A notion of diagnosability is formulated for $\Acal^{\frakDp}$.
	By developing a notion of \emph{concurrent composition}, a necessary and
	sufficient condition is given for diagnosability of automaton $\Acal^{\frakDp}$. It is proven that 
	the problem of computing the concurrent composition for an automaton $\Acal^{\underQ}$ is $\NP$-complete,
	then the problem of verifying diagnosability of $\Acal^{\underQ}$ is proven to be $\coNP$-complete, where the 
	$\NP$-hardness and $\coNP$-hardness results even
	hold for deterministic, deadlock-free, and divergence-free automaton $\Acal^{\underN}$, where $\underQ$
	and $\underN$ are the max-plus dioids having elements in $\Q\cup\{-\infty\}$ and $\N\cup\{-\infty\}$, respectively.
	Several extensions of the main results have also been obtained.



\section{Introduction}

\subsection{Background}

The study of \emph{diagnosability} of discrete-event systems (DES's) modeled 
by finite-state machines originates from
\cite{Lin1994DiagnosabilityDES,Sampath1995DiagnosabilityDES},
where in the former state-based diagnosis was studied,
and in the latter event-based diagnosis was investigated and a \emph{formal definition} of diagnosability
was given. A DES usually consists of discrete states and transitions between states caused by
spontaneous occurrences of partially-observed (labeled) events. 
Intuitively, a DES is called diagnosable if the occurrence of a faulty event can be determined after a
sufficiently large number of subsequent events occur 
by observing the generated output sequence. 
Widely studied models in DES's have finitely many events, but may have infinitely many states, e.g.,
Petri nets, timed automata, max-plus systems, etc. The event-based notion of diagnosability can be
naturally extended to these models, although the characterization of
diagnosability is much more difficult than in finite-state machines.

\subsection{Literature review}
\label{sec:LiterRev}

In the following, we briefly review diagnosability results in different models.

\emph{Labeled finite-state automata.}
In labeled finite-state automata (LFSAs), a notion of diagnosability was formulated in 
\cite{Sampath1995DiagnosabilityDES},
and a \emph{diagnoser}\footnote{a slight variant of the classical powerset construction used for
determinizing a nondeterministic finite automaton \cite{RabinScott1959PowersetConstruction}}
method was proposed to verify diagnosability. The diagnoser of an LFSA records state estimates
along observed output sequences and also records fault propagation along transitions of states of the 
LFSA. The diagnoser has exponential complexity, and diagnosability is verifiable by a relatively simple cycle 
condition on the diagnoser. Hence diagnosability can be verified in $\EXPTIME$. Later, a \emph{twin-plant}
structure and a \emph{verifier} structure both with polynomial complexity were proposed in
\cite{Jiang2001PolyAlgorithmDiagnosabilityDES} and \cite{Yoo2002DiagnosabiliyDESPTime}, respectively,
so that polynomial-time algorithms for verifying diagnosability were given.
The verification algorithms in the above three papers all depend on two fundamental assumptions, namely,
deadlock-freeness/liveness (an LFSA will always run) and divergence-freeness (the running of an LFSA will always
be eventually observed, or, more precisely, for every generated infinitely long transition sequence, its output sequence
is also of infinite length). Later on, in many papers, verification of all kinds of variants of diagnosability
depends on the two assumptions. The reason for why the tools of twin-plant 
  \cite{Jiang2001PolyAlgorithmDiagnosabilityDES} and verifier \cite{Yoo2002DiagnosabiliyDESPTime} do not work
correctly without the two assumptions of deadlock-freeness and divergence-freeness was shown in 
\cite[Section~2.3.4]{Zhang2023DESbook}.
The two assumptions were removed in \cite{Cassez2008FaultDiagnosisStDyObser}
by using a generalized version of the twin-plant structure to verify the \emph{negation} of diagnosability
in polynomial time. The idea of verifying the negation of diagnosability was also independently found in 
\cite{Berard2018DiagnosabilityPetriNet,Zhang2021UnifyingDetDiagPred}.
Results on diagnosability of probabilistic LFSAs can be found in 
\cite{Thorsley2005DiagnosabilityStochDES,Thorsley2017EquivConditionDiagnosabilityStochDES,Keroglou2019AA-DiagnosabilityProbAutomata}, etc.
Results on decentralized settings of diagnosability can be found in \cite{Debouk2000CodiagnosabilityAutomata,Cassez2012ComplexityCodiagnosability,Zhang2021UnifyingDetDiagPred}, etc.
We refer the reader to \cite{Basilio2021ResilienceDES} for a recent survey.
The technique of removing the two assumptions by
adding an unobservable self-loop to each dead state has been widely used in diagnosability studies
in LFSAs, e.g., in \cite{Qiu2006DecentralizedFD,Viana2019CodiagnosabilityDES,Basilio2021ResilienceDES} 
and the references
therein, but this technique works only when the original definition of diagnosability proposed in 
\cite{Sampath1995DiagnosabilityDES} is changed, and the required change was clearly shown in 
\cite{Qiu2006DecentralizedFD} (see Definition~1 and Remark~1 of \cite{Qiu2006DecentralizedFD}). However,
on \cite[Page~355]{Viana2019CodiagnosabilityDES},
it was claimed that the above technique
of adding an unobservable self-loop to each dead state
can remove the two assumptions
without indicating that, for this, the definition of diagnosability needs to be changed.
In Remark~3.20 of \cite{Zhang2021UnifyingDetDiagPred}, it was shown
why this technique does not work for the original definition of diagnosability in \cite{Sampath1995DiagnosabilityDES}.
In \cite{Moreira2011Codiagnosability}, a decentralized version of diagnosability for LFSAs was verified
under the above two assumptions in $\EXPTIME$, but it was wrongly claimed that their algorithm runs in polynomial
time. A polynomial-time verification algorithm is indeed unlikely to exist, because the verification problem 
of the decentralized version of diagnosability is $\PSPACE$-hard \cite{Cassez2012ComplexityCodiagnosability}.

\emph{Labeled Petri nets.}
For results on diagnosability of labeled Petri nets, we refer the reader to 
\cite{Cabasino2012DiagnosabilityPetriNet,Berard2018DiagnosabilityPetriNet,Yin2017DiagnosabilityLabeledPetriNets}, etc.
In \cite{Cabasino2012DiagnosabilityPetriNet}, a new technique called \emph{verifier} (which can be regarded
as the twin-plant structure extended to labeled Petri nets) was developed to
verify diagnosability, and two notions of diagnosability were verified by using the technique under
several assumptions. It was also pointed out that the two notions are equivalent in LFSAs but not 
in labeled Petri nets. In \cite{Yin2017DiagnosabilityLabeledPetriNets}, the weaker notion of
diagnosability studied in \cite{Cabasino2012DiagnosabilityPetriNet} was proven
to be decidable with an $\EXPSPACE$ lower bound under the first of the two previously mentioned 
assumptions, by using the verifier and Yen's path formulae 
\cite{Yen1992YenPathLogicPetriNet,Atig2009YenPathLogicPetriNet};
in \cite{Berard2018DiagnosabilityPetriNet},
an even weaker notion of diagnosability
called trace diagnosability was proven to be decidable in $\EXPSPACE$ with an $\EXPSPACE$ lower bound
without any assumption, by using the verifier and linear temporal logic.
Results on diagnosability in special classes (e.g., bounded, unobservable-event-induced-subnets
being acyclic) of labeled Petri
nets can be found in \cite{Giua2005FaultDetectionPetriNets,Cabasino2005FaultDetectionPetriNets,Basile2009DiagnosisPetriNetsILP}, etc.

\emph{Labeled timed automata.} 
Diagnosability was first defined for labeled timed automata in \cite{Tripakis2002DiagnosisTimedAutomata},
and its decision 
problem was proven to be $\PSPACE$-complete, where $\PSPACE$-membership was proven by
computing a \emph{parallel composition} (which falls back on the generalized twin-plant structure of LFSAs used in
\cite{Cassez2008FaultDiagnosisStDyObser})
in time polynomial in the size of the labeled timed automaton and searching Zeno runs of the parallel composition
in $\PSPACE$ in the size of the parallel composition. The $\PSPACE$ lower bound was obtained by doing
reduction from the $\PSPACE$-complete reachability problem (cf. \cite{Alur1994TimedAutomaton}) in timed automata.
Moreover, since the reachability problem is already $\PSPACE$-hard in $2$-clock timed automata 
\cite{Fearnley2015Reachability2ClocksTimeAutomataPSPACE}, the diagnosability verification problem is 
$\PSPACE$-hard in labeled $2$-clock timed automata.
The diagnoser of a labeled timed automaton was defined
(which may be a Turing machine)
and fault diagnosis was done in $2$-$\EXPTIME$ in the size of the labeled timed automaton and in the size of 
the observation. In \cite{Bouyer2005FaultDiagnosisTimedAutomata}, it was studied whether a labeled timed automaton
has a diagnoser that can be realized by a deterministic timed automaton (or particularly an event
recording automaton). The former was proven to be $2$-$\EXPTIME$-complete and the latter 
$\PSPACE$-complete, provided that a bound on the resources available to the diagnoser is given.

  \emph{Labeled max-plus automata.} In \cite{Lai2022DiagnosabilityMPAutomata}, the authors verified diagnosability
  for labeled unambiguous\footnote{in which under every event sequence, from the
  initial state set to any state there is at most one path, more general than deterministic.}
  max-plus automata over dioid $\underQ$ \eqref{eqn25_diag_LMPautomata} that are 
  deadlock-free and divergence-free in a real-time manner\footnote{which means that the operator $\max$ is 
  not used when computing weighted words, which is the same as the manner considered in the current paper}
  based on the verifier proposed in \cite{Yoo2002DiagnosabiliyDESPTime},
  which is a rather special case of the results on diagnosability of
  general labeled max-plus automata over $\underQ$
  shown in Theorem~\ref{thm3_diag_LMPautomata}. In addition, there is a fundamental mistake in
  the complexity claim in \cite[Remark~2]{Lai2022DiagnosabilityMPAutomata}: This remark claimed a polynomial-time
  verification algorithm for diagnosability, but the algorithm runs in exponential time. In detail, the construction
  of $G^o$ in \cite[Algorithm~2]{Lai2022DiagnosabilityMPAutomata} needs exponential time (in the second for-loop),
  although the size of $G^o$ is polynomial.
  On the other hand, we proved that the diagnosability verification problem for labeled deterministic max-plus automata
  over $\underQ$ that are deadlock-free and divergence-free is $\coNP$-hard \cite{Zhang2021DiagWeightedAutoMonoidCDC}
  (also see Theorem~\ref{thm3_diag_LMPautomata} for a stronger result),
  so unless $\PTIME=\coNP$\footnote{It is widely conjectured that
  $\PTIME\subsetneq\coNP$.}, there is no polynomial-time verification algorithm for diagnosability of such automata,
  let alone for the automata studied in \cite{Lai2022DiagnosabilityMPAutomata}.

\subsection{Contribution of the paper}

In this paper, we consider a labeled weighted automaton $\Acal^{\frakD}$ over a dioid $\frakD$ 
and characterize a notion of diagnosability.

\begin{enumerate}
  \item 
	We consider a special class 
	    of dioids which we call \emph{progressive}. Intuitively, in this class, the canonical order is total,
		there is at least one \emph{eventually dominant} element, there is no zero divisor, 
		and the cancellative law is satisfied. The functionality of an eventually dominant element $t$ is 
		to make every nonzero element $a$ arbitrarily large by multiplying $a$ by $t$
		for sufficiently many times. The role of an eventually dominant element makes a progressive dioid
		exhibit the feature of time elapsing. In addition, progressive dioids have other physical meanings such as
		position deviations.
		The classical max-plus dioids such as the tropical semiring $\underR$,
		$\underQ=(\Q\cup\{-\infty\},\max,+,-\infty,0)$, and $\underN=(\N\cup\{-\infty\},\max,+,-\infty,0)$,
		are all progressive, but their extensions to vectors may not, which depend on the $\oplus$ operator
		(which is specialized into the $\max$ operator in max-plus dioids). The differences caused by different 
		specifications of the $\oplus$ operator are quite subtle. See Section~\ref{sec:dioids} for details.
  \item In order to represent time elapsing, we avoid the standard algebraic semantics of weighted
		automata over semirings, i.e., we consider all generated runs but not only the maximal (corresponding to the 
		$\oplus$ operator in a semiring) runs among the runs with the same event sequence. However, operator $\oplus$
		will not be missing in our paper. The canonical order induced by $\oplus$ will be used throughout the paper by
		defining time elapsing in a generalized sense. The minor overlap between our semantics and the standard
		algebraic semantics of weighted
		automata lies in labeled unambiguous  max-plus automata. 
		On the other hand, we emphasize that 
		the transitions of our automata allow negative weights
	    (see Fig.~\ref{fig6_diag_LMPautomata}), hence it does not fall within the scope of timed automata.
	  \item For a labeled weighted automaton $\Acal^{\frakDp}$ over a progressive dioid $\frakDp$,
		we give a formal definition of diagnosability (Definition~\ref{def_diag_LMPautomata}).
		In order to conveniently represent the definition, we choose runs but not traces (i.e., generated event
		sequences) as in almost all results in the literature (see Section~\ref{sec:LiterRev}).
	  \item By developing a new notion of \emph{concurrent composition} (Definition~\ref{def_CC_diag_LMPautomata},
		originally proposed and computed in \cite{Zhang2022DetWAMonoid}, \cite{Zhang2023DESbook})
		for automaton $\Acal^{\frakDp}$, 
		we give a necessary and sufficient
		condition for (the negation of) diagnosability of $\Acal^{\frakDp}$
		(Theorem~\ref{thm1_diag_LMPautomata}).
	\item Particularly, we prove that the problem of computing concurrent composition for an automaton 
		$\Acal^{\underQ}$ over dioid $\underQ$ is $\NP$-complete (Theorem~\ref{thm4_diag_LMPautomata} and
		Corollary~\ref{cor2_diag_LMPautomata}). The problem of verifying diagnosability
		of $\Acal^{\underQ}$ is then proven to be $\coNP$-complete (Theorem~\ref{thm3_diag_LMPautomata}),
		where the $\NP$-hardness and $\coNP$-hardness results even hold for deterministic, deadlock-free,
		and divergence-free automaton $\Acal^{\underN}$. $\coNP$-membership has been extended 
		to automaton $\Acal^{\underline{\Q^k}}$ over vector progressive dioid 
	    $\underline{\Q^k}$~\eqref{eqn24_diag_LMPautomata} (see Section~\ref{sec:diagextension}).
		
		$\coNP$-membership is obtained by proving 
		that a concurrent composition can be computed in $\NP$ in the size of $\Acal^{\underQ}$
		by connecting $\Acal^{\underQ}$ and the $\NP$-complete \emph{exact path length problem}
		\cite{Nykanen2002ExactPathLength},
		and the negation of diagnosability can be verified in time polynomial in the size of the concurrent
		composition. $\coNP$-hardness is obtained by constructing a polynomial-time reduction from
		the $\NP$-complete \emph{subset sum problem} \cite{Garey1990ComputerIntractability} to the negation 
		of diagnosability.
\end{enumerate}

The remainder of the paper is structured as follows. In Section~\ref{sec:pre}, we introduce
notation, basic facts, the exact path length problem, and labeled 
weighted automata over dioids. We also propose the notion of progressive dioid and prove several useful
properties in such dioids.
We show the main results in Section~\ref{sec:Diag_LWAD} (formulation of the definition of diagnosability
for labeled weighted automata over progressive dioids),
Section~\ref{sec:ConCom_LWAD} (formulation of the basic tool -- concurrent composition for such automata and 
a special case for which the concurrent composition is computable in $\NP$),
and Section~\ref{sec:Verification_LWAD} (a necessary and sufficient condition for diagnosability
based on the concurrent composition and a special algorithmically implementable case in $\coNP$).
Section~\ref{sec:conc} ends up with a short conclusion.

\section{Preliminaries}
\label{sec:pre}

\subsection{Notation}

Symbol $\emptyset$ denotes the empty set.
Symbols $\N$, $\Z$, $\Z_{+}$, $\Q$, $\Q_{+}$, and $\R$ denote the sets of nonnegative integers (natural numbers),
integers, positive integers, rational numbers, positive rational numbers, and real numbers,
respectively. For a finite \emph{alphabet} $\Sig$,
$\Sig^*$ denotes the set of finite \emph{strings} over $\Sig$ including the empty string $\epsilon$.
Elements of $\Sig$ are called \emph{letters}. For $s\in \Sig^*$ and $\sigma
\in\Sig$, we write $\sigma\in s$ if $\sigma$ appears in $s$. $\Sig^{+}:=\Sig^*\setminus\{\epsilon\}$.
For a string $s\in \Sig^*$, $|s|$ stands for its length.
For $s\in \Sig^+$ and $k\in\N$, $\last(s)$ denotes its last letter,
$s^k$ denotes the concatenation of $k$ copies of $s$.
For a string $s\in \Sig^*$, a string $s'\in \Sig^*$ is called a \emph{prefix} of $s$, denoted as $s'\sqsubset s$, 
if there exists another string $s''\in \Sig^*$ such that $s=s's''$.
For $s\in\Sig^*$ and $s'\in\Sig^*$, we use $s\sqsubsetneq s'$ to denote 
$s\sqsubset s'$ and $s\ne s'$. For a string $s\in\Sig^*$, a string $s'\in\Sig^*$ is called a \emph{suffix} of
$s$  if $s''s'=s$ for some $s''\in\Sig^*$.
For two real numbers $a$ and $b$ with $a\le b$, $[a,b]$ denotes the closed interval with lower and 
upper endpoints being $a$ and $b$, respectively;
for two integers $i,j$ with $i\le j$, $\llb i,j\rrb$ denotes the set of all integers no less than
$i$ and no greater
than $j$; for a set $S$, $|S|$ denotes its cardinality and $2^S$ its power set;
$\subset$ and $\subsetneq$ denote the subset and strict subset relations.
$A^{\Tr}$ denotes the transpose of matrix $A$.


\subsection{The exact path length problem}
\label{sec:EPLProblem}

Consider a $k$-dimensional weighted directed graph $G=(\Q^k,V,A)$, where $k\in\Z_{+}$, 
$\Q^k=\underbrace{\Q\times\cdots\times\Q}_{k}$,
$V$ is a finite set of vertices, $A\subset V\times\Q^k\times V$ a finite set of weighted edges with
weights in $\Q^k$. For a path $v_1\xrightarrow[]{z_1}\cdots\xrightarrow[]{z_{n-1}}v_n$, its weight 
is defined by $\sum_{i=1}^{n-1}z_i$. The EPL problem \cite{Nykanen2002ExactPathLength}
is stated as follows.
\begin{problem}[EPL]\label{prob1_det_MPautomata}
	Given a positive integer $k$, a $k$-dimensional weighted directed graph $G=(\Q^k,V,A)$,
	two vertices $v_1,v_2\in V$, and a vector $z\in \Q^k$, determine
	whether there is a path from $v_1$ to $v_2$ with weight $z$.
\end{problem}

\begin{lemma}[\cite{Nykanen2002ExactPathLength}]\label{lem1_det_MPautomata}
	The EPL problem belongs to $\NP$. The EPL problem is $\NP$-hard already for graph $(\Z,V,A)$. 
\end{lemma}

\subsection{Dioids}
\label{sec:dioids}

\begin{definition}\label{dioid_diag_LMPautomata}
An \emph{idempotent semiring} (or \emph{dioid}) is a set $T$ with two binary operators $\oplus$ and $\otimes$,
called \emph{addition} and \emph{multiplication}, such that $(T,\oplus)$ is an \emph{idempotent commutative 
monoid} with \emph{identity element} ${\bf0}\in T$ (also called \emph{zero}), $(T,\otimes)$ is a monoid with 
\emph{identity element} ${\bf1}\in T$ (also called \emph{one}), $\bf0$ is absorbing for $\otimes$,
and $\otimes$ distributes over $\oplus$ on both sides.
A dioid is denoted by $\frakD=(T,\oplus,\otimes,{\bf0},{\bf1})$.
\end{definition}
For $a\in T$, we write $a^n=\underbrace{a\otimes\cdots\otimes a}_{n}$ for all $n\in\N$,
and denote $a^0={\bf1}$.
An \emph{order} over a set $T$ is a relation $\preceq\subset T\times T$ that is reflexive, anti-symmetric,
and transitive.
For all $a,b\in T$ such that $a\preceq b$, we also write $b\succeq a$. If additionally $a\ne b$,
then we write $a\prec b$ or $b\succ a$.

In dioid $\frakD$, the relation $\preceq\subset T\times T$ such that
$a\preceq b$ if and only if $a\oplus b=b$ for all $a,b\in T$ is an order (called
the \emph{canonical order}). In the sequel, $\preceq$ always means the canonical order.
Order $\preceq$ is \emph{total} if for all $a,b\in T$,
either $a\preceq b$ or $b\preceq a$. 

We have some direct properties as follows.

\begin{lemma}\label{lem1_diag_LMPautomata}
	Let $\frakD=(T,\oplus,\otimes,{\bf0},{\bf1})$ be a dioid.
	\begin{enumerate}[(1)]
		\item\label{item4_diag_LMPautomata} 
			Let $b\in T$ be such that ${\bf 1}\preceq b$. Then for all $a\in T$,
			$a\preceq a\otimes b$, $a\preceq b\otimes a$.
		\item\label{item5_diag_LMPautomata}
			Let $a,b\in T$ be such that ${\bf 1}\prec a$ and ${\bf 1}\preceq b$.
			Then ${\bf 1}\prec a\otimes b$ and ${\bf 1}\prec b\otimes a$.
		\item\label{item15_diag_LMPautomata}
			Let $a,b,c,d\in T$ be such that $a\preceq b$ and $c\preceq d$. Then $a\otimes c\preceq
			b\otimes d$.
		\item\label{item14_diag_LMPautomata}
			Let $a,b\in T$ be such that ${\bf 1}\prec a\prec b$.
			Then $a\prec a\otimes b$ and $a\prec b\otimes a$.
		\item\label{item14'_diag_LMPautomata}
			Let $a,b\in T$ be such that ${\bf 1}\succ a\succ b$.
			Then $a\succ a\otimes b$ and $a\succ b\otimes a$.
	\end{enumerate}
\end{lemma}

%
%
%
%

The results in Lemma~\ref{lem1_diag_LMPautomata} show that a dioid exhibits some features
of \emph{time elapsing}, if we consider an element greater than ${\bf1}$ as a positive time elapsing
(${\bf1}$ means no time elapses),
and also consider the product of two elements greater than ${\bf1}$ as the total time elapsing (see
\eqref{item14_diag_LMPautomata}).
In order to see whether a dioid has more features of time elapsing, one may ask
given $a\in T$ such that ${\bf 1}\prec a$, for a sufficiently large $n\in\N$, whether (\romannumeral1) $a^n$
can be arbitrarily large, and moreover, whether (\romannumeral2) $a^n$ can make any nonzero element $t$ 
{arbitrarily large via multiplying $t$ by $a^n$. Generally not. Consider the following example.
\begin{example}
Consider dioid 
\begin{equation}\label{eqn5_diag_LMPautomata} 
	(\underQ)^k:=
	\left(\left(\Q\cup\{-\infty\}\right)^k,\max,+,\begin{bmatrix}-\infty\\\vdots\\-\infty\end{bmatrix},
	\begin{bmatrix}0\\\vdots\\0\end{bmatrix}\right),
\end{equation}
where $(-\infty)+r=r+(-\infty)=-\infty$ for all $r\in\Q$, $+$ is defined componentwise on the vectors of
$\left(\Q\cup\{-\infty\}\right)^k$; for every two distinct vectors $[a_1,\dots,a_k]^{\Tr}$ and 
$[b_1,\dots,b_k]^{\Tr}$, $[a_1,\dots,a_k]^{\Tr}\prec[b_1,\dots,b_k]^{\Tr}$ (equivalently, $[b_1,\dots,b_k]^{\Tr}
\\\succ[a_1,\dots,a_k]^{\Tr}$) if and only if
$a_i<b_i$ for some $i$ and $a_j=b_j$ for all $j<i$; for all $u,v\in \left(\Q\cup\{-\infty\}\right)^k$,
$\max\{u,v\}$ is equal to the larger one of $u$ and $v$ if they are not equal, equal to $u$ if they are equal.
$\preceq$ is the canonical order of $(\underQ)^k$,
and also the well-known lexicographic order and hence total.

Consider vector $v:=[1,0,\dots,0]^{\Tr}$. $v^n
=[n,0,\dots,0]^{\Tr}$ can be arbitrarily large. Formally, for every vector $
u\in\left(\Q\cup\{-\infty\}\right)^k$, $v^m\succ u$ for sufficiently large $m$. That is, the above (\romannumeral1) holds.
However, the above (\romannumeral2) does not hold. Choose $[-\infty,0,\dots,0]^{\Tr}=:w$, $w+v^n\not\succ
[0,\dots,0]^{\Tr}$ for any $n\in\N$.

It is not difficult to see that no vector in $(\Q\cup\{-\infty\})^k$ satisfies both the above
(\romannumeral1) and (\romannumeral2).
\end{example}

By choosing a subdioid\footnote{i.e., a dioid $\left(Y,\max,+,\begin{bmatrix}-\infty\\\vdots\\-\infty
\end{bmatrix},
\begin{bmatrix}0\\\vdots\\0\end{bmatrix}\right)$, where 
$Y\subset \left(\Q\cup\{-\infty\}\right)^k$} of $(\underQ)^k$~\eqref{eqn5_diag_LMPautomata}, we can find a vector
that satisfies both (\romannumeral1) and (\romannumeral2).
\begin{example}
	Consider dioid 
	\begin{equation}\label{eqn24_diag_LMPautomata} 
		\underline{\Q^k}:=
		\left(\Q^k\cup\left\{\begin{bmatrix}-\infty\\\vdots\\-\infty\end{bmatrix}\right\},\max,+,
		\begin{bmatrix}-\infty\\\vdots\\-\infty\end{bmatrix},
		\begin{bmatrix}0\\\vdots\\0\end{bmatrix}\right)
	\end{equation}
	as a subdioid of $(\underQ)^k$~\eqref{eqn5_diag_LMPautomata}. Now we again choose $v=[1,0,\dots,0]^{\Tr}$.
	This $v$ satisfies the above (\romannumeral1) and (\romannumeral2). 
\end{example}
Informally, we call elements in a dioid that satisfy (\romannumeral1) and (\romannumeral2) eventually dominant. Such elements are
formulated as follows.

\begin{definition}\label{def3_diag_LMPautomata}
	Let $\frakD=(T,\oplus,\otimes,{\bf0},{\bf1})$ be a dioid. An element $t\in T$ is called \emph{eventually
	dominant} if for all ${\bf0}\ne a\in T$, there exist $m,n,p\in\Z_{+}$ such that $t^m\otimes a\succ{\bf1}$,
	$a\otimes t^n\succ{\bf1}$, and $t^p\succ a$. Analogously, an element $t\in T$ is called \emph{eventually
	anti-dominant} if for all ${\bf0}\ne a\in T$, there exist $m,n,p\in\Z_{+}$ such that $t^m\otimes a\prec{\bf1}$,
	$a\otimes t^n\prec{\bf1}$, and $t^p\prec a$. The set of eventually dominant elements and the set of eventually
	anti-dominant elements are denoted by $T_{\ED}$ and $T_{\EaD}$, respectively.
\end{definition}

By Lemma~\ref{lem1_diag_LMPautomata}, the following results hold.

\begin{lemma}\label{lem7_diag_LMPautomata}  
	Let $\frakD=(T,\oplus,\otimes,{\bf0},{\bf1})$ be a dioid with a total canonical order. Let
	$u,v\in T$ be eventually dominant and eventually anti-dominant, respectively. Let ${\bf1}\prec a_1,
	\dots,a_n\in T$ be not eventually dominant, where $n\ge 1$. Let ${\bf1}\succ b_1,\dots,b_m\in T$
	be not eventually anti-dominant, where $m\in\Z_{+}$. The following hold.
	\begin{enumerate}
		\item\label{item30_diag_LMPautomata}
			$u\succ{\bf1}$, $v\prec{\bf1}$.
		\item\label{item31_diag_LMPautomata}
			$u\prec u^2\prec u^3\prec\cdots$, $v\succ v^2\succ v^3\succ\cdots$.
		\item\label{item32_diag_LMPautomata}
			$a_1\prec u$, $a_1\otimes\cdots\otimes a_n\prec u$.
		\item\label{item33_diag_LMPautomata}
			$b_{1}\succ v$, $b_{1}\otimes\cdots\otimes b_{m}\succ v$.
	\end{enumerate}
\end{lemma}

  \begin{proof}
	We only need to prove item~\eqref{item32_diag_LMPautomata}, item~\eqref{item33_diag_LMPautomata} holds similarly.
	If $a_1\succeq u$ then by definition $a_1$ is eventually dominant. Then since the canonical order is total, 
	$a_1\prec u$. By item~\eqref{item15_diag_LMPautomata} of Lemma~\ref{lem1_diag_LMPautomata}, we have
	$a_1\otimes\cdots \otimes a_n\preceq (\max\{a_1,\dots,a_n\})^n$. Based on the same argument we have 
	$(\max\{a_1,\dots,a_n\})^n\prec u$, otherwise $\max\{a_1,\dots,a_n\}$ is eventually dominant. 
	Then $a_1\otimes \dots\otimes a_n\prec u$.
  \end{proof}

\begin{remark}\label{rem1_diag_LMPautomata}
	Lemma~\ref{lem7_diag_LMPautomata} implies that in a dioid with a total canonical order, the product
	of any finitely many elements that are not eventually dominant is less than any
	eventually dominant element; similarly, the product of any finitely many elements that are not eventually
	anti-dominant is greater than any eventually anti-dominant element.
\end{remark}

\begin{example}\label{exam6_diag_LMPautomata}
	In dioid $(\underQ)^k$~\eqref{eqn5_diag_LMPautomata}, there is no eventually dominant element. In 
	dioid $\underline{\Q^k}$~\eqref{eqn24_diag_LMPautomata}, all vectors $[a_1,\dots,a_k]^{\Tr}$ with
	$a_1>0$ are exactly the eventually dominant elements, none of the other vectors greater than $[0,\dots,0]^{\Tr}$
	(e.g., $[0,1,\dots,0]^{\Tr}$) is eventually dominant. Similarly, the eventually anti-dominant elements
	are exactly the vectors $[a_1,\dots,a_k]^{\Tr}$ with $a_1<0$. Now we illustrate
	item~\eqref{item32_diag_LMPautomata} of Lemma~\ref{lem7_diag_LMPautomata}. Consider vectors 
	$[0,a_i,\dots,0]^{\Tr}$, $i\in\llb 1,n \rrb$, with $a_i>0$, greater than $[0,\dots,0]^{\Tr}$ but not
	eventually dominant. Their sum is $[0,\sum_{i=1}^na_i,\dots,0]^{\Tr}$ which is less than any eventually
	dominant element.
\end{example}

}

Apart from the above (\romannumeral1) and (\romannumeral2), we additionally consider 
the \emph{cancellative law}: for all $a,b,c\in T$ with
$a\ne{\bf0}$, $a\otimes b=a\otimes c$ $\implies$ $b=c$ and $b\otimes a=c\otimes a$ $\implies$ $b=c$.
A dioid with the cancellative law is called \emph{cancellative}. For example, dioid
$(\underQ)^k$~\eqref{eqn5_diag_LMPautomata}
is not cancellative, because $[-\infty,0,\dots,0]^{\Tr}+[1,0,\dots,0]^{\Tr}=
[-\infty,0,\dots,0]^{\Tr}+[2,0,\dots,0]^{\Tr}$, but $[1,0,\dots,0]^{\Tr}\ne[2,0,\dots,0]^{\Tr}$.
While its subdioid $\underline{\Q^k}$~\eqref{eqn24_diag_LMPautomata} is cancellative.

\begin{lemma}\label{lem6_diag_LMPautomata} 
	Consider a cancellative dioid $\frakD=(T,\oplus,\otimes,{\bf0},\\{\bf1})$. For all $a,b,c,d\in T$, the following hold.
	\begin{enumerate}[(1)]
		\item\label{item19_diag_LMPautomata} 
			If $a\otimes b\succeq a\otimes c$ and $a\ne{\bf0}$, then $b\succeq c$.
		\item\label{item20_diag_LMPautomata}
			If $c\otimes d\succeq a\otimes b$, $c\preceq a$, and $a\ne{\bf0}$, then $d\succeq b$.
		\item\label{item29_diag_LMPautomata}  
			If $a\ne{\bf0}$ and $b\succ c$ then $a\otimes b\succ a\otimes c$ and $b\otimes a\succ c\otimes a$.
	\end{enumerate}
\end{lemma}

With these preliminaries, the definition of \emph{progressive dioid} is formulated as follows.

\begin{definition}\label{progressive_dioid_diag_LMPautomata}
	A dioid $\frakD=(T,\oplus,\otimes,{\bf0},{\bf1})$ is called \emph{progressive} if (\romannumeral1) 
	the canonical order of $\frakD$ is total,
	(\romannumeral2) $\frakD$ has at least one eventually dominant element, (\romannumeral3) 
	$\frakD$ is cancellative, (\romannumeral4) $\frakD$ has no zero divisor (i.e., for all nonzero $a,b$ in $T$,
	$a\otimes b\ne{\bf0}$). A progressive dioid is denoted by $\frakDp$.
\end{definition}

For example, dioid $(\underQ)^k$~\eqref{eqn5_diag_LMPautomata} has zero divisors. For example,
$[-\infty,1,-\infty,\dots,-\infty]^{\Tr}$ and $[1,-\infty,\\-\infty,\dots,-\infty]^{\Tr}$ are zero divisors
because 
$[-\infty,1,-\infty,\dots,-\infty]^{\Tr}+[1,-\infty,-\infty,\dots,-\infty]^{\Tr}=[-\infty,-\infty,-\infty,\dots,-\infty]^{\Tr}$.
While its subdioid $\underline{\Q^k}$~\eqref{eqn24_diag_LMPautomata} has no zero divisor.

\begin{remark}
	By \eqref{item31_diag_LMPautomata} of Lemma~\ref{lem7_diag_LMPautomata}, one sees in a progressive dioid $\frakDp$, a maximal element cannot  
	exist, that is, $\oplus_{t\in T}t\notin T$. That is, a progressive dioid cannot be complete.
	Item~\eqref{item31_diag_LMPautomata}
	of Lemma~\ref{lem7_diag_LMPautomata}
	also implies that progressive dioids must have infinitely many elements.
\end{remark}

One sees that the dioids
\begin{align}
  \underQ &:= (\Q\cup\{-\infty\},\max,+,-\infty,0),\label{eqn25_diag_LMPautomata}\\
	\undernonnegQ &:= (\Q_{\ge 0}\cup\{-\infty\},\max,+,-\infty,0),\label{eqn26_diag_LMPautomata}\\
	\underN &:= (\N\cup\{-\infty\},\max,+,-\infty,0),\label{eqn27_diag_LMPautomata}\\
	\underR &:= (\R\cup\{-\infty\},\max,+,-\infty,0)\label{eqn28_diag_LMPautomata}
\end{align}
are progressive and have a total canonical order $\le$. Dioid
$(\underQ)^k$~\eqref{eqn5_diag_LMPautomata} is not progressive,
because it only satisfies (\romannumeral1) of Definition~\ref{progressive_dioid_diag_LMPautomata}.
Dioid $\underline{\Q^k}$~\eqref{eqn24_diag_LMPautomata} is progressive.  In dioid~$\underQ$,
positive numbers in $\Q$ are exactly the eventually dominant elements in $\underQ$, then there is no 
element that is greater than the one element of $\underQ$ (i.e., the standard $0$) but not eventually dominant,
but in dioid $\underline{\Q^k}$~\eqref{eqn24_diag_LMPautomata}, such an element exists (see 
Example~\ref{exam6_diag_LMPautomata}). Later in Section~\ref{sec:diagextension} (see
Example~\ref{exam7_diag_LMPautomata}),
a practical example of an autonomous robot is shown to illustrate
the practical use of labeled weighted automata over vector progressive dioid~\eqref{eqn24_diag_LMPautomata}.

\subsection{Labeled \texorpdfstring{$\mathfrak{D}$}{D}-automata}

A \emph{weighted automaton} is a tuple $\frakG=(Q,E,\Delta,Q_0,\alpha,\mu)$ over dioid
$\frakD=(T,\oplus,\otimes,{\bf0},{\bf1})$, denoted by 
	$(\frakD,\frakG)$
for short, where $Q$ is a nonempty finite set of \emph{states},
$E$ a nonempty finite \emph{alphabet} (elements of $E$ are called \emph{events}),
$\Delta\subset Q\times E\times Q$ a \emph{transition relation} (elements of $\Delta$ are 
called \emph{transitions}), $Q_0\subset Q$ is a nonempty set of initial states,
$\alpha$ is a map $Q_0\to T\setminus\{ {\bf0}\}$, 
map $\mu:\Delta\to T\setminus\{ {\bf0}\}$ assigns to each transition $(q,e,q')\in\Delta$ a nonzero
\emph{weight} $\mu(e)_{qq'}$ in $T$, 
where this transition is also denoted by $q\xrightarrow[]{e/\mu(e)_{qq'}}q'$. 
For a self-loop $q\xrightarrow[]{e/\mu(e)_{qq}}q$,
we use $q\left(\xrightarrow[]{e/\mu(e)_{qq}}q\right)^t$ to denote the concatenation of 
$q$ and $t$ copies of $\xrightarrow[]{e/\mu(e)_{qq}}q$, where $t\in\N$.
Automaton $(\frakD,\frakG)$
is called \emph{deterministic} if $|Q_0|=1$ and for all $q,q',q''\in Q$ and $e\in E$, $(q,e,q')\in\Delta$
and $(q,e,q'')\in\Delta$ $\implies$ $q'=q''$. For all $q\in Q$, we also regard
$q\xrightarrow[]{\ep/{\bf1}}q$ as a transition (which is called an \emph{$\ep$-transition}).
A transition $q\xrightarrow[]{e/\mu(e)_{qq'}}q'$ is called \emph{instantaneous} if $\mu(e)_{qq'}={\bf1}$
(recall that $\bf1$, the one element of $\frakD$, is specialized into the standard $0$ of dioid $\underQ$.),
and called \emph{noninstantaneous} if $\mu(e)_{qq'}\ne{\bf1}$.
Call a state $q\in Q$ \emph{dead}
if for all $q'\in Q$ and $e\in E$, $(q,e,q')\not\in\Delta$, i.e., there exists no transition
starting at $q$ (apart from the $\ep$-transition). Call an automaton $(\frakD,\frakG)$ \emph{deadlock-free}/\emph{live}
if it has no reachable dead state.

Particularly for 
	$(\undernonnegQ,\frakG)$,
for an initial state $q\in Q_0$, $\alpha(q)$ denotes its initial
time delay, and in a transition $q\xrightarrow[]{e/\mu(e)_{qq'}}q'$, $\mu(e)_{qq'}$ denotes its time
delay, i.e., the time consumption of the execution of the transition.
Hence the execution of an instantaneous transition requires zero time, 
while the execution of a noninstantaneous transition requires a positive rational number $\mu(e)_{qq'}$ as time.

Consider automaton $(\frakD,\frakG)$.
For $q_0,q_1,\dots,q_n\in Q$ (where $q_0$ is not necessarily initial), and $e_1,\dots,e_n\in E$, 
where $n\in \N$, we call 
\begin{align}\label{path_det_MPautomaton}
	\pi:=q_0\xrightarrow[]{e_1}q_1\xrightarrow[]{e_2}\cdots\xrightarrow[]{e_n}q_n
\end{align}
a \emph{path} if for all $i\in\llb 0,n-1\rrb$, $(q_i,e_{i+1},q_{i+1})\in\Delta$.
We write $\init(\pi)=q_0$, $\last(\pi)=q_n$.
A path $\pi$ is 
called a \emph{cycle} if $q_0=q_n$. A cycle $\pi$ is called \emph{simple} if it has no repetitive states
except for $q_0$ and $q_n$.
For paths $\pi_1$ and $\pi_2$ with $\last(\pi_1)=\init(\pi_2)$,
we use $\pi_1\pi_2$ to denote the concatenation of $\pi_1$ and $\pi_2$ after removing
one of $\last(\pi_1)$ or $\init(\pi_2)$, and write $(\pi_1\pi_2)\setminus \pi_1=\pi_2$.
The set of paths starting at $q_0\in Q$ and ending at $q\in Q$ 
(under event sequence $s:=e_1\dots e_n\in E^*$)
is denoted by $q_0\rightsquigarrow q$ ($q_0\xrsquigarrow{s}q$), where 
$q_0\xrsquigarrow{s}q=\{q_0\xrightarrow[]{e_1}q_1\xrightarrow[]{e_2}\cdots\xrightarrow[]{e_{n-1}} q_{n-1}
\xrightarrow[]{e_n}q|q_1,\dots,q_{n-1}\in Q\}$. Then for $q\in Q$ and $s\in E^*$,
we denote $Q_0\rightsquigarrow q:=\bigcup_{q_0\in Q_0}q_0\rightsquigarrow{}q$ and $Q_0\xrsquigarrow{s}q:=
\bigcup_{q_0\in Q_0}q_0\xrsquigarrow{s}q$. 

The \emph{weighted word} of path $\pi$ is defined by 
\begin{align}\label{timedword_det_MPautomaton}
	\tau(\pi):=(e_1,t_1)(e_2,t_2)\dots(e_n,t_n),
\end{align}
where for all $i\in\llb 1,n\rrb$, 
$t_i=\bigotimes_{j=1}^{i}\mu(e_j)_{q_{j-1}q_j}$. The \emph{weight} $\WEI(\pi)$ (also denoted by
$\WEI_{\pi}$) of path $\pi$ and the \emph{weight} $\WEI(\tau(\pi))$ (also denoted by $\WEI_{\tau(\pi)}$)
of weighted word $\tau(\pi)$ are both defined by $t_n$. 
A path always has nonzero weight if $\frakD$ has no zero divisor.
A path $\pi$ is called \emph{instantaneous} if $t_1=\cdots=t_n={\bf1}$, and called \emph{noninstantaneous}
otherwise. 
Call a state $q\in Q$ \emph{stuck} if either $q$ is dead or starting at $q$
there exist only instantaneous paths. Automaton $(\frakD,\frakG)$ is called \emph{stuck-free} if it has no reachable
stuck state. Intuitively, when automaton $(\undernonnegQ,\frakG)$ is in a stuck state, the automaton may still
be running but no time elapses. 

Particularly for 
$(\undernonnegQ,\frakG)$, one has $t_i=\sum_{j=1}^{i}{\mu(e_j)_{q_{j-1}q_j}}$ (here the operator $\otimes$ 
in a general dioid $\frakD$ is specified as $+$), hence
$t_i$ denotes the time needed for the first $i$ transitions in path
$\pi$, $i\in\llb 1,n\rrb$.

We define a labeling function $\ell:E\to\Sig\cup\{\ep\}$, where $\Sig$ is a finite alphabet,
to distinguish \emph{observable} and \emph{unobservable} events. The set of observable events and the set of
unobservable events are denoted by $E_o=\{e\in E|\ell(e)\in\Sig\}$ and $E_{uo}=\{e\in E|\ell(e)=\ep\}$,
respectively. When an observable event $e$ occurs, one observes $\ell(e)$; while an unobservable
event occurs, one observes nothing. A transition $q\xrightarrow[]{e/\mu(e)_{qq'}}q'$ is called
\emph{observable} (resp., \emph{unobservable}) if $e$ is observable (resp., unobservable).
Labeling function $\ell$ is recursively extended to
$E^*\to \Sig^*$ as $\ell(e_1e_2\dots e_n)=\ell(e_1)\ell(e_2)\dots\ell(e_n)$.
A path $\pi$ \eqref{path_det_MPautomaton} is called \emph{unobservable} if $\ell(e_1\dots e_n)=\ep$,
and called \emph{observable} otherwise.  

A \emph{labeled weighted automaton over dioid $\frakD$} is formulated as 
\begin{align}\label{LMPA_diag_opa_MPautomata}
	\Acal^{\frakD}:=(\frakD,\frakG,\Sig,\ell),
\end{align}
and is also denoted \emph{labeled $\frakD$-automaton}.

Labeling function $\ell$ is extended as follows:
for all $(e,t)\in E\times T$, $\ell((e,t))=(\ell(e),t)$ if $\ell(e)\ne\ep$, and $\ell((e,t))=\ep$
otherwise. Hence $\ell$ is also recursively extended to $(E\times T)^*\to
(\Sig\times T)^*$. For a path $\pi$, $\ell(\tau(\pi))$
is called a \emph{weighted label/output sequence}.
We also extend the previously defined function
$\tau$ as follows: for all $\gamma=(\s_1,t_1')\dots(\s_n,t_n')\in(\Sig\times T)^{*}$, 
\begin{equation}\label{eqn11_det_MPautomata} 
	\tau(\gamma)=(\s_1,t_1)\dots(\s_n,t_n),
\end{equation}
where $t_j=\bigotimes_{i=1}^{j}t_i'$ for all $j\in\llb 1,n\rrb$. Particularly, $\tau(\ep)=\ep$.


From now on, without loss of generality, we assume for each initial state $q_0\in Q_0$, 
$\alpha(q_0)={\bf1}$, because otherwise we can add a new initial state $q_0'$ not in $Q_0$
and set $\alpha(q_0')={\bf1}$,
and for each initial state $q_0\in Q_0$ such that $\alpha(q_0)\ne{\bf1}$ we add a new transition
$q_0'\xrightarrow[]{\varepsilon/\alpha(q_0)}q_0$ and set $q_0$ to be not initial any more,
where $\varepsilon$ is a new event not in $E$. For consistency we also
assume that $\ell(\varepsilon)=\ep$.

Particularly for $\Acal^{\undernonnegQ}$, if it generates a path
$\pi$ as in \eqref{path_det_MPautomaton} such that $q_0\in Q_0$, consider its timed word $\tau(\pi)$ as in 
\eqref{timedword_det_MPautomaton}, then at time $t_i$, one will observe $\ell(e_i)$ if $\ell(e_i)\ne\ep$;
and observe nothing otherwise, where $i\in\llb 1,n\rrb$.

\begin{example}\label{exam2_diag_LMPautomata}
	A stuck-free labeled $\underN$-automaton $\Acal_1^{\underN}$ is shown in Fig.~\ref{fig3_diag_LMPautomata}. 
	\begin{figure}[!htbp]
	\centering
	\begin{tikzpicture}
	[>=stealth',shorten >=1pt,thick,auto, node distance=2.3 cm, scale = 0.8, transform shape,
	->,>=stealth,inner sep=2pt, initial text = 0]

	\tikzstyle{emptynode}=[inner sep=0,outer sep=0]

	\node[initial, state, initial where = above] (q0) {$q_0$};
	\node[state] (q2) [right of = q0] {$q_2$};
	\node[state] (q1) [left of = q0] {$q_1$};
	\node[state] (q3) [left of = q1] {$q_3$};
	\node[state] (q4) [right of = q2] {$q_{4}$};

	\path [->]
	(q0) edge node [above, sloped] {${\cyan f}/3$} (q1)
	(q0) edge node [above, sloped] {$u/1$} (q2)
	(q1) edge [loop above] node [above, sloped] {$u/1$} (q1)
	(q2) edge [loop above] node [above, sloped] {$u/1$} (q2)
	(q1) edge node [above, sloped] {$a/1$} (q3)
	(q2) edge node [above, sloped] {$a/1$} (q4)
	(q3) edge [loop above] node {$a/1$} (q3)
	(q4) edge [loop above] node {$a/1$} (q4)
	;

     \end{tikzpicture}
	 \caption{A labeled $\underN$-automaton $\Acal_1^{\underN}$, where dioid $\underN$ is shown in 
	 \eqref{eqn27_diag_LMPautomata}, $q_0$ is the
	 unique initial state, $a$ is the observable event, $u,{\cyan f}$ are unobservable events, i.e., $\ell(a)=a$, 
	 $\ell(u)=\ell({\cyan f})=\ep$. $\cyan f$ is the unique faulty event. The natural numbers after ``$/$'' 
	 denote the weights of the corresponding transitions.}  
	 \label{fig3_diag_LMPautomata}
	 \end{figure}
\end{example}


\section{The definition of diagnosability}
\label{sec:Diag_LWAD}

In this section, we formulate the definition of diagnosability. 
The set of \emph{faulty} events is denoted by $E_f\subset E$. Usually, one assumes
that $E_f\subset E_{uo}$ \cite{Sampath1995DiagnosabilityDES,Tripakis2002DiagnosisTimedAutomata,Jiang2001PolyAlgorithmDiagnosabilityDES},
because the occurrence of observable events can be directly seen.
However, technically, this assumption is usually not needed. We do not make this assumption in this paper,
because different observable events may have the same label (i.e., labeling function $\ell$ is not necessarily
injective on $E_{o}$), so the occurrences of observable faulty events may also need to be distinguished.
A labeling function was used as early as 1975 \cite{Hack1975PetriNetLanguage} and was also
called a mask in \cite{Cieslak1988ObservabilityDES}. As usual, with slight abuse of notation,
for $s\in E^*$, we use $E_f\in s$ to denote that some faulty event $e_f\in E_f$ appears in $s$.
Similarly for $w\in (E\times T)^*$, we use $E_f\in w$ to denote that some $(e_f,t)$ appears in $w$, 
where $e_f\in E_f$,
$t\in T$. In $\Acal^{\frakD}$, all transitions of the form $q_1\xrightarrow[]{e_f/\mu(e_f)_{q_1q_2}}
q_2$ for some $e_f\in E_f$ are called \emph{faulty transitions}, the other transitions are called \emph{normal
transitions}. 
We denote by $\Acal^{\frakD}_{\nsf}$ the \emph{normal subautomaton} of $\Acal^{\frakD}$ obtained by
removing all faulty transitions of $\Acal^{\frakD}$ and subsequently removing all unreachable states and all
their ingoing and outgoing transitions. We also denote by $\Acal^{\frakD}_{\fsf}$ the \emph{faulty
subautomaton} of $\Acal^{\frakD}$ obtained by only keeping all reachable
faulty transitions and their predecessors and successors,
where predecessors mean the transitions from which some faulty transition is reachable, and successors 
mean the transitions that are reachable from some faulty transition.

\begin{example}
	Reconsider the automaton $\Acal_1^{\underN}$ in Example~\ref{exam2_diag_LMPautomata} (shown in 
	Fig.~\ref{fig3_diag_LMPautomata}). Its faulty subautomaton $\Acal_{1\fsf}^{\underN}$ and normal subautomaton
	$\Acal_{1\nsf}^{\underN}$ are shown in Fig.~\ref{fig4_diag_LMPautomata}.
	\begin{figure}[!htbp]
	\centering
	 \begin{tikzpicture}
	[>=stealth',shorten >=1pt,thick,auto, node distance=2.5 cm, scale = 0.8, transform shape,
	->,>=stealth,inner sep=2pt, initial text = 0]

	\tikzstyle{emptynode}=[inner sep=0,outer sep=0]

	\node[initial, state, initial where = above] (q0') {$q_0$};
	\node[state] (q2') [right of = q0'] {$q_2$};
	\node[state] (q4') [right of = q2'] {$q_{4}$};

	\path [->]
	(q0') edge node [above, sloped] {$u/1$} (q2')
	(q2') edge [loop above] node [above, sloped] {$u/1$} (q2')
	(q2') edge node [above, sloped] {$a/1$} (q4')
	(q4') edge [loop above] node {$a/1$} (q4')
	;

	\node[initial, state, initial where = above, above of = q4'] (q0'') {$q_0$};
	\node[state] (q1'') [left of = q0''] {$q_1$};
	\node[state] (q3'') [left of = q1''] {$q_3$};

	\path [->]
	(q0'') edge node [above, sloped] {${\cyan f}/3$} (q1'')
	(q1'') edge [loop above] node [above, sloped] {$u/1$} (q1'')
	(q1'') edge node [above, sloped] {$a/1$} (q3'')
	(q3'') edge [loop above] node {$a/1$} (q3'')
	;

     \end{tikzpicture}
	 \caption{$\Acal_{1\fsf}^{\underN}$ (above) and $\Acal_{1\nsf}^{\underN}$ (below) corresponding to
	 automaton $\Acal_1^{\underN}$ in Fig.~\ref{fig3_diag_LMPautomata}.}
	 \label{fig4_diag_LMPautomata}
	 \end{figure}
\end{example}

\begin{definition}\label{def_diag_LMPautomata}
  Let $\frakDp$ be a progressive dioid, $\Acal^{\frakDp}$ a labeled $\frakDp$-automaton, 
  and 
	$E_f\subset E$ a set of faulty events. $\Acal^{\frakDp}$ is called \emph{$E_f$-diagnosable}
	(see Fig.~\ref{fig2_diag_LMPautomata} for illustration) if
	\begin{equation}\label{eqn17_diag_LMPautomata}
	\begin{split} 
		&(\exists t\in T_{\ED})(\forall \pi\in Q_0 \xrsquigarrow{se_f} q\text{ with }e_f\in E_f)\\
		&(\forall \pi' \in q\xrsquigarrow{s'} q')[ ( (\WEI_{\pi'} \succ t) \vee (q'\text{ is stuck}) ) \implies
		{\bf D}],
	\end{split}
	\end{equation}
	where ${\bf D}=(\forall \pi''\in Q_0\xrsquigarrow{s''}q''\text{ with }\ell(\tau(\pi''))=\ell(\tau(\pi\pi')))
	[((\WEI_{\pi''}\succeq \WEI_{\pi}\otimes t)\vee(q''\text{ is stuck}))\implies (E_f\in s'')].$
\end{definition}

\begin{figure}[!htbp]
	\begin{center}
				\centering
		\begin{tikzpicture}[>=stealth',shorten >=1pt,auto,node distance=2.4 cm, scale = 1.5, transform shape,
	>=stealth,inner sep=2pt,
		empty/.style={}]

		\draw [draw=black, rounded corners=3] (0.0,0.0) rectangle (5.0,2.0);

		\draw [red] (0.0,1) -- (2.0,1);
		\draw [green] (2.0,1) -- (5.0,1);

		\draw [dashed, blue] (0.0,1.7) .. controls (3.0,1.8) .. (5.0,1.9);
		\draw [dashed, blue] (0.0,0.8) .. controls (3.0,0.7) .. (5.0,0.6);
		\draw [dashed, blue] (0.0,0.6) .. controls (3.0,0.5) .. (4.0,0.4);

		\filldraw[cyan] (2,1) circle (2pt) node[anchor=south] {\scriptsize$e_f\in E_f$};
		\filldraw[cyan] (1,1.74) circle (2pt);
		\filldraw[cyan] (3,0.7) circle (2pt);
		\filldraw[cyan] (3.5,0.45) circle (2pt);
		\draw[brown] (4.0,0.4) circle (2pt);

		\draw [dotted] (2.0,0) -- (2.0,-0.5);
		\draw [dotted] (5.0,0) -- (5.0,-0.5);
		\node at (3.5,-0.25) {\scriptsize$\WEI_{\green\pi'}\succ t$};
		\node at (3.5,1.15) {\scriptsize${\green\pi'}$};
		\node at (1.0,1.15) {\scriptsize${\red\pi}$};
		\draw [->] (4.3,-0.25) -- (5.0,-0.25);
		\draw [->] (2.70,-0.25) -- (2.0,-0.25);

		\end{tikzpicture}
	\end{center}
	\caption{Illustration of diagnosability. Each line represents a path starting from an initial state,
	where the weighted label sequences of all dashed lines are the same as that of the solid line at the instant
	when the last event of the solid line occurs. Bullets denote faulty events. The circle denotes a stuck state.}
	\label{fig2_diag_LMPautomata}
\end{figure}

Definition~\ref{def_diag_LMPautomata} can be interpreted as follows. For a progressive dioid $\frakDp$,
automaton $\Acal^{\frakDp}$ is $E_f$-diagnosable if and only if there exists an eventually dominant element 
$t$ in $T$, for every path $\pi$ whose last event is faulty,
for every path $\pi'$ as a continuation of $\pi$, if
either $\WEI_{\pi'}\succ t$ or $\pi'$ ends at a stuck state, then one can 
make sure that a faulty event (although not necessarily $e_f$) must have occurred when the weighted label sequence 
$\ell(\tau(\pi\pi'))$ has been generated. 

Particularly, $E_f$-diagnosability of $\Acal^{\undernonnegQ}$ has the following practical meaning:
if $\Acal^{\undernonnegQ}$ is $E_f$-diagnosable, then once a faulty event (e.g., $e_f$) occurs
at some instant (e.g., $\WEI_\pi$), then after a sufficiently
long time (e.g., $t\in \Q_+$) that only depends on $\Acal^{\undernonnegQ}$,
one can infer that some faulty event must have occurred by observing the generated 
timed label sequence (e.g., $\ell(\tau(\pi\pi'))$) from time $0$ to $\WEI_{\pi}+t$.

We can simplify the investigation of $E_f$-diagnosability of automaton
$\Acal^{\frakDp}$ by modifying $\Acal^{\frakDp}$ as follows: 
at each stuck state, a normal,
noninstantaneous, and unobservable transition to a \emph{sink} state is added, and on the sink state,
a normal, noninstantaneous, and unobservable self-loop is added.
Hence the labeled stuck-free $\frakDp$-automaton
\begin{equation}\label{SF_LMPA_diag_opa_LMPautomata}
	\bar\Acal^{\frakDp}=(\frakDp,\bar\frakG,\Sig,\bar\ell)
\end{equation}
is obtained from $\Acal^{\frakDp}$ by modifying $\Acal^{\frakDp}$ as follows:
at each stuck state $q\in Q$, add a transition $q\xrightarrow[]{\frakun/\frakt}\frakq$,
where $\frakun$ is a normal, unobservable event not in $E$, $\frakq$ is an additional state not in $Q$,
$\frakt\in T_{\ED}$, and add a self-loop $\frakq\xrightarrow[]{\frakun/\frakt}\frakq$.
Denote $\bar Q:=Q\cup\{\frakq\}$. Compared with the automaton $\frakG$ in \eqref{LMPA_diag_opa_MPautomata}, 
automaton $\bar\frakG$ has an additional state $\frakq$ and these additionally added transitions from the stuck states 
of $\frakG$ to $\frakq$ and the self-loop on $\frakq$. Compared with the domain of the labeling function $\ell$
in \eqref{LMPA_diag_opa_MPautomata}, the domain of the labeling function $\bar\ell$ has an additional event
$\frakun$.

\begin{lemma}\label{lem5_diag_LMPautomata}
  Let $\frakDp$ be a progressive dioid, $\Acal^{\frakDp}$ a labeled $\frakDp$-automaton, 
  and $E_f\subset E$ a set of faulty events. Let
	$\bar\Acal^{\frakDp}$ be its stuck-free automaton as in \eqref{SF_LMPA_diag_opa_LMPautomata}.
	Then $\Acal^{\frakDp}$ is $E_f$-diagnosable if and only if $\bar\Acal^{\frakDp}$ satisfies
	\begin{equation}\label{eqn20_diag_LMPautomata}
	\begin{split} 
		&(\exists t\in T_{\ED})(\forall \pi\in Q_0 \xrsquigarrow{se_f} q\text{ with }e_f\in E_f)\\
		&(\forall \pi' \in q\xrsquigarrow{s'} \bar Q)[ (\WEI_{\pi'} \succ t)  \implies
		{\bf D}],
	\end{split}
	\end{equation}
	where ${\bf D}=(\forall \pi''\in Q_0\xrsquigarrow{s''}q''\text{ with }\bar\ell(\tau(\pi''))=\bar\ell
	(\tau(\pi\pi')))[(\WEI_{\pi''}\succeq \WEI_{\pi}\otimes t)\implies (E_f\in s'')]$.
\end{lemma}

\begin{proof}
	Observe that in a path $\pi$ of $\bar\Acal^{\frakDp}$, if $\frakun$ appears, then after the $\frakun$ occurs,
	only $\frakq$ can be visited, hence all successors are $\frakq\xrightarrow[]{\frakun/\frakt}\frakq$.
	Without loss of generality,
	we assume that $\Acal^{\frakDp}$ is not stuck-free, otherwise $\Acal^{\frakDp}$ is the same as its stuck-free
	automaton $\bar\Acal^{\frakDp}$.

	``only if'':

	Suppose that $\Acal^{\frakDp}$ is $E_f$-diagnosable, i.e., \eqref{eqn17_diag_LMPautomata} holds.
	Choose a $t\succ {\bf1}$ as in \eqref{eqn17_diag_LMPautomata}. In $\bar\Acal^{\frakDp}$, 
	arbitrarily choose paths $\pi\in Q_0\xrsquigarrow{se_f}q$, $\pi'\in q\xrsquigarrow{s'}q'$, and
	$\pi''\in Q_0\xrsquigarrow{s''}q''$ such that $e_f\in E_f$, $\WEI_{\pi'}\succ t$,
	$\bar\ell(\tau(\pi''))=\bar\ell(\tau(\pi\pi'))$, and $\WEI_{\pi''}\succeq\WEI_{\pi}\otimes t$.
	Let $\bar \pi''\in Q_0\xrsquigarrow{\bar s''}\bar q''$
	be the longest prefix of $\pi''$ that contains no $\frakun$. Then $\bar\pi''$ is a path of $\Acal^{\frakDp}$
	and $\ell(\tau(\bar\pi''))= \bar\ell(\tau(\pi''))$. In addition,
	$\last(\bar\pi'')$ is a stuck state of $\Acal^{\frakDp}$ if $\bar\pi''$ is not the same as $\pi''$.

	(\romannumeral1) Assume $\frakun\notin se_fs'$. Then $\pi\pi'$ is also a path of $\Acal^{\frakDp}$, and hence
	$E_f\in\bar s''$ by \eqref{eqn17_diag_LMPautomata}. It follows that $E_f\in s''$.

	(\romannumeral2) Assume $\frakun\notin se_f$ and $\frakun\in s'$. Then $\pi$ is also a path of $\Acal^{\frakDp}$.
	Let $\bar \pi'$ be the longest prefix of $\pi'$ that contains no $\frakun$, then $\pi\bar\pi'$ is a path of
	$\Acal^{\frakDp}$, $\ell(\tau(\pi\bar\pi')) = \bar\ell(\tau(\pi\pi'))$, and $\last(\bar \pi')$ is a stuck state 
	of $\Acal^{\frakDp}$ if $\bar\pi'$ is not the same as $\pi'$.
	We then have $\ell(\tau(\bar\pi''))=\ell(\tau(\pi\bar\pi'))$. 
	Then by \eqref{eqn17_diag_LMPautomata}, we have $E_f\in \bar s''$. It also follows that $E_f\in s''$.

	(\romannumeral3) By definition of $\bar \Acal^{\frakDp}$, $\frakun\notin se_f$.

	Based on the above (\romannumeral1), (\romannumeral2), and (\romannumeral3), no matter whether 
	$\frakun\in se_fs'$ or not, one has $E_f\in s''$, then \eqref{eqn20_diag_LMPautomata} holds.

	``if'': 

	Suppose $\bar\Acal^{\frakDp}$ satisfies \eqref{eqn20_diag_LMPautomata}. Choose a $t\succ{\bf1}$ as in
	\eqref{eqn20_diag_LMPautomata}. In $\Acal^{\frakDp}$, arbitrarily choose paths
	$\pi\in Q_0\xrsquigarrow{se_f}q$, $\pi'\in q\xrsquigarrow{s'}q'$, and
	$\pi''\in Q_0\xrsquigarrow{s''}q''$ such that $e_f\in E_f$, $\WEI_{\pi'}\succ t$ or $q'$ is stuck, 
	$\ell(\tau(\pi''))=\ell(\tau(\pi\pi'))$, $\WEI_{\pi''}\succeq\WEI_{\pi}\otimes t$ or $q''$ is stuck.

	One directly sees that $\pi\pi'$ and $\pi''$ are paths of $\bar\Acal^{\frakDp}$. If $\WEI_{\pi'}\nsucc t$,
	then $q'$ is stuck, we denote $\bar\pi':=\pi'\left( \xrightarrow[]{\frakun}\frakq \right)^{2n}$, where
	$n\in\Z_+$ is such that $\WEI_{\pi'}\otimes \frakt^n\succ {\bf1}$ and
	$\frakt^n\succ t$ ($n$ exists because $\frakt$ is eventually dominant,
	$\frakt=\mu(\frakun)_{\frakq \frakq}\succ{\bf1}$
	(by Lemma~\ref{lem7_diag_LMPautomata}),
	then $\WEI_{\bar\pi'}=\WEI_{\pi'}\otimes \frakt^{2n}\succeq
	\frakt^n \succ t$ by Lemma~\ref{lem1_diag_LMPautomata}. Similarly,
	if $\WEI_{\pi''}\nsucceq\WEI_{\pi}\otimes t$ then $q''$ is stuck, we can add sufficiently many 
	$\frakun$-transitions after $q''$ so that the updated $\pi''$ satisfies $\WEI_{\pi''}\succeq \WEI_{\pi}\otimes t$. 
    Then by \eqref{eqn20_diag_LMPautomata}, $E_f\in s''$.
\end{proof}

By Lemma~\ref{lem5_diag_LMPautomata},
in order to investigate whether a labeled $\frakDp$-automaton $\Acal^{\frakDp}$ 
is $E_f$-diagnosable (Definition~\ref{def_diag_LMPautomata}), we only need to investigate whether its
stuck-free automaton $\bar\Acal^{\frakDp}$ satisfies \eqref{eqn20_diag_LMPautomata}. We refer to 
\eqref{eqn20_diag_LMPautomata} as the definition of $E_f$-diagnosability for labeled stuck-free 
$\frakDp$-automaton $\bar\Acal^{\frakDp}$.


\begin{definition}\label{def'_diag_LMPautomata}
	Let $\frakDp$ be a progressive dioid, $\Acal^{\frakDp}$ a labeled stuck-free $\frakDp$-automaton as in 
	\eqref{SF_LMPA_diag_opa_LMPautomata}, and 
	$E_f\subset E$ a set of faulty events. $\Acal^{\frakDp}$ is called \emph{$E_f$-diagnosable} if
	\begin{equation}\label{eqn20'_diag_LMPautomata}
	\begin{split} 
		&(\exists t\in T_{\ED})(\forall \pi\in Q_0 \xrsquigarrow{se_f} q\text{ with }e_f\in E_f)\\
		&(\forall \pi' \in q\xrsquigarrow{s'} Q)[ (\WEI_{\pi'} \succ t)  \implies
		{\bf D}],
	\end{split}
	\end{equation}
	where ${\bf D}=(\forall \pi''\in Q_0\xrsquigarrow{s''}q''\text{ with }\ell(\tau(\pi''))=\ell(\tau(\pi\pi')))
	[(\WEI_{\pi''}\succeq \WEI_{\pi}\otimes t)\implies (E_f\in s'')]$.
\end{definition}

By Definition~\ref{def'_diag_LMPautomata}, one directly has the following result.

\begin{proposition}\label{prop2_diag_LMPautomata}
	Let $\frakDp$ be a progressive dioid. A labeled stuck-free $\frakDp$-automaton $\Acal^{\frakDp}$ 
	\eqref{SF_LMPA_diag_opa_LMPautomata} is not $E_f$-diagnosable if and only if 
	\begin{equation}\label{eqn10_diag_LMPautomata}
		\begin{split}
			&(\forall t\in T_{\ED})(\exists \pi_t\in Q_0 \xrsquigarrow{se_f} q\text{ with }e_f\in E_f)\\
			&(\exists \pi_t' \in q\xrsquigarrow{s'} Q)
			(\exists \pi_t''\in Q_0\xrsquigarrow{s''}q'')[(\WEI_{\pi_t'}\succ t) \wedge \\
			&(\ell(\tau(\pi_t''))=\ell(\tau(\pi_t\pi_t'))) \wedge (\WEI_{\pi_t''}\succeq\WEI_{\pi}\otimes t)\\
			&\wedge (E_f\notin s'')].
		\end{split}
	\end{equation}
\end{proposition}

In words, a labeled stuck-free $\frakDp$-automaton is not $E_f$-diagnosable if and only if
there exist two runs from initial states at the same time such that they produce the same weighted label 
sequence, the first run
contains a faulty event and after the event the weight is arbitrarily large, the second run
contains no faulty event and its weight is also arbitrarily large.

\begin{example}\label{exam3_diag_LMPautomata}
	Recall the automaton $\Acal_1^{\underN}$ in Example~\ref{exam2_diag_LMPautomata} (shown in 
	Fig.~\ref{fig3_diag_LMPautomata}). We choose paths
	\begin{align*}
		 {\red\pi_t} &:= q_0 \xrightarrow[]{{\cyan f}/3} q_1,
		 \quad {\green\pi_t'} := q_1 \xrightarrow[]{a/1}q_3\left( \xrightarrow[]{a/1}q_3 \right)^t,\\
		 {\blue\pi_t''} &:= q_0 \left( \xrightarrow[]{u/1}q_2  \right) ^3 
		 \left( \xrightarrow[]{a/1}q_4 \right)^{t+1},
	 \end{align*}
	 where $t\in\Z_+$. It holds that 
	 \begin{align*}
		 &\ell(\tau({\blue\pi_t''})) = \ell(\tau({\red\pi_t}{\green\pi_t'})) = (a,4)(a,5)\dots(a,t+4).
	 \end{align*}

	 Then $\Acal_1^{\underN}$ is not $\cyan\{f\}$-diagnosable by Proposition.~\ref{prop2_diag_LMPautomata}.
\end{example}

\section{The notion of concurrent composition}
\label{sec:ConCom_LWAD}

In this section, for a labeled $\frakD$-automaton, we define a notion of \emph{concurrent composition} 
of its faulty subautomaton $\Acal^{\frakD}_{\fsf}$ and its normal subautomaton $\Acal^{\frakD}_{\nsf}$,
and will use the concurrent composition to derive a necessary and sufficient condition for 
the negation of diagnosability of labeled $\frakDp$-automata.
In \cite{Zhang2022DetWAMonoid}, the notion of concurrent composition was originally defined for a
labeled weighted automaton $\Acal^{\frakM}$ over a monoid $\frakM$ and computed for a labeled weighted automaton 
$\Acal^{\Q^k}$ over the monoid $(\Q^k,+)$ in $\NP$ by developing original methods, resulting in that the 
classical strong detectability results in labeled finite-state automata 
\cite{Shu2007Detectability_DES} were for the first time extended to 
$\Acal^{\Q^k}$. The concurrent composition used in the current paper is a variant of the original version
proposed in \cite{Zhang2022DetWAMonoid} adapted to $\Acal^{\frakD}$. Compared with the technique of verifying 
strong detectability proposed in \cite{Zhang2022DetWAMonoid}, the technique of verifying diagnosability
used in the current paper is more involved because of the additional consideration of orders in dioids.
From now on, without loss of generality $\Acal^{\frakD}$ is always assumed stuck-free.

\begin{definition}\label{def_CC_diag_LMPautomata}
	Consider a labeled $\frakD$-automaton $\Acal^\frakD$ \eqref{LMPA_diag_opa_MPautomata}
	and a faulty event set $E_f\subset E$.
	We define the \emph{concurrent composition} of $\Acal^{\frakD}_{\fsf}$ and $\Acal^{\frakD}_{\nsf}$
	as the LFSA
\begin{align}\label{CC_diag_LMPautomata}
  \CCa(\Acal^{\frakD}_{\fsf},\Acal^{\frakD}_{\nsf})=(Q',E',\dt',Q_0',\Sig,\ell'),
\end{align}
where $Q'=Q\times Q$; $E'=E'_{o}\cup E'_{uo}$, where $E'_{o}=\{(e_1,e_2)\in E_o\times (E_o\setminus
E_f)|\ell(e_1)=\ell(e_2)\}$;
$E'_{uo}=\{(e_{uo},\epsilon)|e_{uo}\in E_{uo}\}\cup\{(\epsilon,e_{uo})|e_{uo}\in 
E_{uo}\setminus E_f\}$; $Q_0'=Q_0\times Q_0$;
$\dt'\subset Q'\times E'\times Q'$ is the transition relation, for all states $(q_1,q_2),(q_3,q_4)\in Q'$,
and events $(e_1,e_2)\in E_o'$, $(e_{uo}^1,\ep),(\ep,e_{uo}^2)\in E_{uo}'$,
\begin{enumerate}[(1)]
	\item\label{item1_diag_LMPautomata}
	  (observable transition) $((q_1,q_2),(e_1,e_2),(q_3,q_4))\in\dt'$ if and only if in $\Acal^{\frakD}$, there exist states $q_5,q_6\in
Q$, event sequences $s_1\in (E_{uo})^*$, $s_2\in (E_{uo}\setminus E_f)^*$, and paths
	\begin{equation}\label{eqn6_diag_LMPautomata}
		\begin{split}
			\pi_1 &:= q_1\xrightarrow[]{s_1}q_5\xrightarrow[]{e_1}q_3,\\
			\pi_2 &:= q_2\xrightarrow[]{s_2}q_6\xrightarrow[]{e_2}q_4,
		\end{split}
	\end{equation}
	in $\Acal^{\frakD}_{\fsf}$, $\Acal^{\frakD}_{\nsf}$, respectively, such that $\WEI_{\pi_1}=\WEI_{\pi_2}$
	(particularly when $\frakD=\underQ$, we additionally define the \emph{weight} of this transition by $0$, because 
	$\WEI_{\pi_1}=\WEI_{\pi_2}$),
	\item\label{item2_diag_LMPautomata}
		(unobservable transition) $((q_1,q_2),(e_{uo}^1,\ep),(q_3,q_4))\in\dt'$ if and only if 
		$(q_1,e_{uo}^1,q_3)\in\Delta$, $q_2=q_4$ (particularly when $\frakD=\underQ$,
		we additionally define the \emph{weight} of the transition by $\mu(e_{uo}^1)_{q_1q_3}$),
	\item\label{item3_diag_LMPautomata}
		(unobservable transition) $((q_1,q_2),(\ep,e_{uo}^2),(q_3,q_4))\in\dt'$ if and only if $q_1=q_3$,
		$(q_2,e_{uo}^2,q_4)\in\Delta$ (particularly when $\frakD=\underQ$, we additionally define 
		the \emph{weight} of the transition by $-\mu(e_{uo}^2)_{q_2q_4}$);
\end{enumerate}
for all $(e_1,e_2)\in E_o'$, $\ell'((e_1,e_2))=\ell(e_1)$; for all $e'\in E_{uo}'$, $\ell'(e')=\ep$.
$\ell'$ is recursively extended to $(E')^*\to \Sig^*$. 
For a state $q'$ of $\CCa(\Acal^{\frakD}_{\fsf},\Acal^{\frakD}_{\nsf})$, we write $q'=(q'(L),q'(R))$\footnote{Note
that $\CCa(\Acal^{\underQ}_{\fsf},\Acal^{\underQ}_{\nsf})$ is a labeled $\underQ$-automaton.
Although here we only define weights for transitions of $\CCa(\Acal^{\underQ}_{\fsf},
\Acal^{\underQ}_{\nsf})$, we want to emphasize that weights can be defined for transitions of each 
$\CCa(\Acal^{\frakD}_{\fsf}, \Acal^{\frakD}_{\nsf})$, provided
every nonzero element of $\frakD$ has a multiplicative inverse, e.g., see the vector progressive
dioid~\eqref{eqn24_diag_LMPautomata} and Example~\ref{exam7_diag_LMPautomata}.}
where $q'(L)$
and $q'(R)$ denote the left component and right component of $q'$, respectively. Such notation also applies
to event sequences and transition sequences consisting of unobservable transitions in $\CCa(\Acal^{\frakD}_{\fsf},
\Acal^{\frakD}_{\nsf})$.
\end{definition}

Intuitively, the original idea of constructing the concurrent composition $\CCa(\Acal^{\frakD}_{\fsf},
\Acal^{\frakD}_{\nsf})$ comes from characterizing the negation of $E_f$-diagnosability as in
Proposition~\ref{prop2_diag_LMPautomata}. $\CCa(\Acal^{\frakD}_{\fsf},\Acal^{\frakD}_{\nsf})$ collects all pairs of 
runs $\pi_t\pi_t'$ and $\pi_t''$ of $\Acal^{\frakD}$ as in \eqref{eqn10_diag_LMPautomata}, i.e., the weighted label
sequence of $\pi_t\pi_t'$ is the same as that of $\pi_t''$, the last event of $\pi_t$ is faulty, the weights
of $\pi_t'$ and $\pi_t''$ are both arbitrarily large. On the other hand, $\CCa(\Acal^{\frakD}_{\fsf},\Acal^{\frakD}_{\nsf})$
is an LFSA. Particularly, since in the dioid $\underQ$, every nonzero element $t$ (i.e., $t\ne-\infty$) has a
multiplicative inverse $-t$ (the multiplication in $\underQ$ is the standard $+$), a transition of 
$\CCa(\Acal^{\underQ}_{\fsf},\Acal^{\underQ}_{\nsf})$ is endowed with a weight as the difference between the weight
of its left component and that of its right component.

One sees that in $\CCa(\Acal^{\frakD}_{\fsf},\Acal^{\frakD}_{\nsf})$, an observable transition 
is obtained by merging two paths of $\Acal^{\frakD}$, where the two paths contain exactly
one observable event each, have the same weight and both end at the occurrences of the 
observable events; in addition, the second path contains only normal transitions.
Hence the two paths are consistent with observations. Particularly for $\Acal^{\undernonnegQ}$, in both
paths, the observable events occur at the same instant of time if $\Acal^{\undernonnegQ}$ starts at the starting 
states of the two paths at the same instant.
However, an unobservable transition of $\CCa(\Acal^{\frakD}_{\fsf},\Acal^{\frakD}_{\nsf})$ is obtained by merging an
unobservable transition and an $\ep$-transition of $\Acal^{\frakD}$. The two transitions are 
not necessarily consistent with observations, e.g., an unobservable transition may have weight $\ne{\bf1}$, 
but an $\ep$-transition must have weight equal to $\bf1$.
A sequence $q_0'\xrightarrow[]{s_1'}\cdots\xrightarrow[]{s_n'}q_n'$
of transitions of $\CCa(\Acal^{\frakD}_{\fsf},\Acal^{\frakD}_{\nsf})$ ($n\in\N$) is called a \emph{run}.
An \emph{unobservable run} is a run consisting of unobservable transitions.
In an unobservable run $\pi'$, its left component
$\pi'(L)$ is an unobservable path of $\Acal^{\frakD}_{\fsf}$, its right component $\pi'(R)$
is an unobservable path of $\Acal^{\frakD}_{\nsf}$. We denote by $\CCao(\Acal^{\frakD}_{\fsf},\Acal^{\frakD}_{\nsf})$ and
$\CCauo(\Acal^{\frakD}_{\fsf},\Acal^{\frakD}_{\nsf})$ the subautomata of $\CCa(\Acal^{\frakD}_{\fsf},\Acal^{\frakD}_{\nsf})$
consisting of all reachable observable transitions and all reachable unobservable transitions,
respectively.


Computing all unobservable transitions of
$\CCa(\Acal^{\underQ}_{\fsf},\Acal^{\underQ}_{\nsf})$ takes time polynomial in the size of
$\Acal^{\underQ}$. 
Now for states $(q_1,q_2),(q_3,q_4)\in Q'$ and observable event $(e_1,e_2)\in E_o'$,
we check whether there exists an observable transition $((q_1,q_2),(e_1,e_2),(q_3,q_4))$
in $\CCa(\Acal^{\underQ}_{\fsf},\Acal^{\underQ}_{\nsf})$:
\begin{myenumerate}
	\item\label{item16_diag_LMPautomata}
		Nondeterministically choose $q_5,q_6\in Q$ such that
		$q_5\xrightarrow[]{e_1} q_3$ and $q_6\xrightarrow[]{e_2} q_4$ are two observable transitions of
		$\Acal^{\underQ}_{\fsf}$ and $\Acal^{\underQ}_{\nsf}$, respectively.
	\item\label{item17_diag_LMPautomata}
		Regard $\CCauo(\Acal^{\underQ}_{\fsf},\Acal^{\underQ}_{\nsf})$ as a $1$-dimensional weighted 
		directed graph $(\Q,V,A)$ as in Subsection~\ref{sec:EPLProblem}, where $V=Q'$, $A=\{((p_1,p_2),
		\mu(e_{uo})_{p_1p_3},(p_3,p_2))|p_1,p_2,\\p_3\in Q,e_{uo}\in E_{uo},(p_1,e_{uo},p_3)\in\Delta \}
		\cup \{((p_1,p_2),\\-\mu(e_{uo})_{p_2p_3},(p_1,p_3))|p_1,p_2,p_3\in Q,e_{uo}\in E_{uo},\\(p_2,e_{uo},
		p_3)\in\Delta\}$.
	\item\label{item18_diag_LMPautomata}
		If there exists a path from $(q_1,q_2)$ to $(q_5,q_6)$ with weight $\mu(e_2)_{q_6q_4}-\mu(e_1)
		_{q_5q_3}$ in $\CCauo(\Acal^{\underQ}_{\fsf},\Acal^{\underQ}_{\nsf})$,
		then $((q_1,q_2),(e_1,e_2),(q_3,q_4))$ is an observable transition of
		$\CCa(\Acal^{\underQ}_{\fsf},\Acal^{\underQ}_{\nsf})$.
\end{myenumerate}

By Lemma~\ref{lem1_det_MPautomata}, the existence of 
such a path can be checked in $\NP$. Hence the following result holds.

\begin{theorem}\label{thm4_diag_LMPautomata}
	Consider a labeled $\underQ$-automaton $\Acal^{\underQ}$, a faulty event set $E_f\subset E$,
	and the corresponding faulty subautomaton $\Acal^{\underQ}_{\fsf}$ and normal subautomaton
	$\Acal^{\underQ}_{\nsf}$. The concurrent composition $\CCa(\Acal^{\underQ}_{\fsf},\Acal^{\underQ}_{\nsf})$
	can be computed in time nondeterministically polynomial in the size of $\Acal^{\underQ}$.
\end{theorem}

In order to characterize $E_f$-diagnosability, we define several special kinds of transitions 
and runs in $\CCa(\Acal^{\frakD}_{\fsf},\Acal^{\frakD}_{\nsf})$.

\begin{definition}\label{def1_diag_LMPautomata} 
	Consider an observable transition $(q_1,q_2)\\\xrightarrow[]{(e_1,e_2)}(q_3,q_4)=:\pi_o$ of $\CCa(\Acal^
	{\frakD}_{\fsf},\Acal^{\frakD}_{\nsf})$ as in 
	Definition~\ref{def_CC_diag_LMPautomata} \eqref{item1_diag_LMPautomata}, a path as $\pi_1$ or $\pi_2$
	in \eqref{eqn6_diag_LMPautomata} is called an \emph{admissible path} of $\pi_o$.
	Transition $\pi_o$ is called \emph{faulty} if it has an admissible path $\pi_1$ containing a
	faulty event, called \emph{dominant} if it has an admissible path $\pi_1$ such that $\WEI_{\pi_1}$
	is eventually dominant.
	A run $q_0'\xrightarrow[]{e_1'}\cdots\xrightarrow[]{e_n'}q_n'$ of $\CCao(\Acal^{\frakD}_{\fsf},\Acal^{\frakD}_{\nsf})$ 
	is called
	\emph{dominant} if each transition $q_i'\xrightarrow[]{e_{i+1}'}q_{i+1}'$ has an admissible path
	$\pi_i$ such that $\bigotimes_{i=0}^{n-1}\WEI_{\pi_i}$ is eventually dominant.
	Hence a dominant observable transition is a dominant run, a run 
	consisting of dominant observable transitions is a dominant run by 
	Lemma~\ref{lem1_diag_LMPautomata}.
	An unobservable transition of $\CCa(\Acal^{\frakD}_{\fsf},\Acal^{\frakD}_{\nsf})$ is \emph{faulty} if
	its event belongs to $E_f\times\{\ep\}$.
\end{definition}

In order to check in $\CCao(\Acal^{\underQ}_{\fsf},\Acal^{\underQ}_{\nsf})$, whether an observable
transition $\pi_o$
is faulty, we only need to add a condition $E_f\in s_1e_1$ into \eqref{item18_diag_LMPautomata}
so that we need to solve a subproblem of the EPL problem. Hence whether 
$\pi_o$ is faulty can be checked in $\NP$ in the size of 
$\Acal^{\underQ}$. By definition, $\pi_o$ is dominant if and only if it admits an admissible path whose 
weight is positive, because positive numbers in $\Q$ are exactly the eventually dominant elements in $\underQ$.
Then similarly, whether $\pi_o$ is dominant can also be checked in $\NP$.
Consider a simple cycle $q_0'\xrightarrow[]{e_1'}\cdots\xrightarrow[]{e_n'}q_n'$
of $\CCao(\Acal^{\underQ}_{\fsf},\Acal^{\underQ}_{\nsf})$, where $q_0'=q_n'$. In order to check whether
the cycle is dominant, we can make $n$ copies of $\CCauo(\Acal^{\underQ}_{\fsf},\Acal^{\underQ}_{\nsf})$,
and use these copies and $\Acal^{\underQ}$ to check whether for each $i\in\llb 0,n-1 \rrb$, $q_i'\xrightarrow[]
{e_{i+1}'}q_{i+1}'$ has an admissible run $\pi_1^i$ such that $\sum_{i=0}^{n-1} \WEI_{\pi_1^i}>0$.
This can be done also in $\NP$, because $n$ must be less than or equal to $|Q|^2$.

\begin{proposition}\label{prop3_diag_LMPautomata}
	In $\CCao(\Acal^{\underQ}_{\fsf},\Acal^{\underQ}_{\nsf})$, whether a transition is faulty, whether 
	a transition is dominant, whether a simple cycle is dominant, can be verified in $\NP$ in the size
	of $\Acal^{\underQ}$.
\end{proposition}

\begin{definition}\label{def2_diag_LMPautomata} 
	In $\Acal^{\frakD}$, a state $q$ is called \emph{dominant} (resp., \emph{anti-dominant})
	if there exists an unobservable cycle $q\xrightarrow[]{s}q=:\pi$ such that $\WEI_{\pi}$ is eventually
	dominant (resp., eventually anti-dominant). Such a cycle is called a
	\emph{dominant cycle} (resp., \emph{anti-dominant cycle}).
	A state $q$ is called \emph{eventually dominant} if either $q$ is dominant
	or some dominant state $q'$ is reachable from $q$ through some unobservable path.
\end{definition}

For $\Acal^{\underQ}$, we next show that all dominant states 
can be computed in time polynomial in the size of $\Acal^{\underQ}$: (1) Compute all reachable unobservable
transitions of $\Acal^{\underQ}$, denote the obtained subautomaton by $\Acc(\Acal^{\underQ}_{uo})$,
which is considered as a weighted directed graph $(\Q,V,A)$.
(2) Compute all strongly connected components of $(\Q,V,A)$ by using Tarjan algorithm,
which takes time linear in the size of $(\Q,V,A)$.
(3) Choose a vertex $q$ of $(\Q,V,A)$ and the strongly connected component 
$SCC_{q}$ containing $q$. In $SCC_{q}$, replace the weight $w$ of any transition
by $-w$, denote the currently updated $SCC_{q}$ by $\overline{SCC_{q}}$.
Then $q$ is dominant if and only if in $\overline{SCC_{q}}$ there is a cycle with
negative weight. One can use Bellman-Ford algorithm \cite[Chap. 24]{Cormen2009Algorithms} to
check whether there is a cycle with negative weight reachable from $q$ in time $O(|V||A|)$.
Similarly, all anti-dominant states can also be computed in polynomial time.

\begin{proposition}\label{prop4_diag_LMPautomata}
	In $\Acal^{\underQ}$, whether a state is dominant (resp., anti-dominant)
	can be verified in time polynomial in the size of $\Acal^{\underQ}$.
\end{proposition}

Furthermore, we need one property of $\CCauo(\Acal^{\frakD}_{\fsf},
\Acal^{\frakD}_{\nsf})$ in order to characterize $E_f$-diagnosability.

\begin{proposition}\label{prop1_diag_LMPautomata}
	Consider a labeled $\frakD$-automaton $\Acal^{\frakD}$. In the corresponding $\CCauo(\Acal^{\frakD}_{\fsf},
	\Acal^{\frakD}_{\nsf})$, for every run
	\begin{equation}\label{eqn1_diag_LMPautomata}
		q_0'\xrightarrow[]{s_1'}\cdots\xrightarrow[]{s_n'}q_n',
	\end{equation}
	where $n\in\Z_{+}$, $s_1',\dots,s_n'\in (E'_{uo})^+$, there exists a run
	\begin{equation}\label{eqn2_diag_LMPautomata}
		q_0'\xrightarrow[]{\bar s_1'}\bar q_1'\xrightarrow[]{\bar s_2'}q_n'
	\end{equation}
	such that $\bar s_1'\in (E_{uo}\times\{\ep\})^*$, $\bar s_2'\in (\{\ep\}\times(E_{uo}\setminus E_f))^*$,
	$q_0'(R)=\bar q_1'(R)$, $\bar q_1'(L)=q_n'(L)$,
	and the left (resp., right) component of \eqref{eqn1_diag_LMPautomata} is the same as 
	the left (resp., right) component of \eqref{eqn2_diag_LMPautomata}.
\end{proposition}

We refer the reader to \cite{Zhang2021DiagWeightedAutoMonoidCDC} for its proof.

\begin{example}
	Reconsider the automaton $\Acal_1^{\underN}$ in Example~\ref{exam2_diag_LMPautomata} (shown in 
	Fig.~\ref{fig3_diag_LMPautomata}). 
	The concurrent composition $\CCa(\Acal_{1\fsf}^{\underN},\Acal_{1\nsf}^{\underN})$
	is shown in Fig.~\ref{fig5_diag_LMPautomata}.
	\begin{figure}[!htbp]
	\centering
	 \begin{tikzpicture}
	[>=stealth',shorten >=1pt,thick,auto, node distance=3.0 cm, scale = 0.8, transform shape,
	->,>=stealth,inner sep=2pt]

	\tikzstyle{emptynode}=[inner sep=0,outer sep=0]

	\node[initial, state, initial where = above] (00) {$q_0q_0$};
	\node[state] (34) [left of = 00] {$q_3q_4$};
	\node[state] (12) [below of = 34] {$q_1q_2$};
	\node[state] (10) [right of = 12] {$q_1q_0$};
	\node[state] (02) [left of = 34] {$q_0q_2$};
	
	\path [->]
	(00) edge node {$({\cyan f},\ep)/3$} (10)
	
	(10) edge node {$(\ep,u)/-1$} (12)
	;

	\path [->]
	
	(12) edge [loop left] node {$\begin{array}{l}(u,\ep)/1 \\ (\ep,u)/-1\end{array}$} (12)
	
	(10) edge [loop below] node {$(u,\ep)/1$} (10)

	(02) edge [loop above] node {$(\ep,u)/-1$} (12)
	(02) edge node [above, sloped] {$({\cyan f},\ep)/3$} (12)
	;

	\path [->]
	(00) edge node [above] {$(a,a)/0$} (34)
	;

	\path [->]
	(12) edge node [above, sloped] {$(a,a)/0$} (34)
	(10) edge node [above, sloped] {$(a,a)/0$} (34)
	(02) edge node [above, sloped] {$(a,a)/0$} (34)
	(34) edge [loop above] node {$(a,a)/0$} (34)
	;

	\draw [->] (00) .. controls  (-3,2.5) .. (02);
	\node at (-3,2.3) {$(\ep,u)/-1$};
    \end{tikzpicture}
	\caption{$\CCa(\Acal_{1\fsf}^{\underN},\Acal_{1\nsf}^{\underN})$ corresponding to
	 automaton $\Acal_1^{\underN}$ in Fig.~\ref{fig3_diag_LMPautomata}.}
	\label{fig5_diag_LMPautomata}
	\end{figure}
\end{example}

\section{{A necessary and sufficient condition for diagnosability and algorithmically implementable results}}
\label{sec:Verification_LWAD}

In this section, we give a necessary and sufficient condition for $E_f$-diagnosability of a labeled
$\frakDp$-automaton, and prove that the $E_f$-diagnosability 
verification problem for $\Acal^{\underline{\Q^k}}$ over vector progressive dioid 
$\underline{\Q^k}$~\eqref{eqn24_diag_LMPautomata} is $\coNP$-complete.

\subsection{A necessary and sufficient condition for \texorpdfstring{$E_f$}{Ef}-diagnosability of \texorpdfstring{$\mathcal{A}^{\mathfrak{D_p}}$}{AD}}

In this subsection, for a labeled $\frakDp$-automaton $\Acal^{\frakDp}$,
we use the corresponding concurrent composition $\CCa(\Acal^{\frakDp}_{\fsf},\Acal^{\frakDp}_{\nsf})$ 
to give a necessary and sufficient condition for its $E_f$-diagnosability (as in 
Definition~\ref{def'_diag_LMPautomata}). By Lemma~\ref{lem5_diag_LMPautomata}, we assume without loss of 
generality that $\Acal^{\frakDp}$ is stuck-free.

\begin{theorem}\label{thm1_diag_LMPautomata}
  A labeled $\frakDp$-automaton $\Acal^{\frakDp}$ 
	is not $E_f$-diagnosable 
	if and only if in the corresponding concurrent composition $\CCa(\Acal^{\frakDp}_{\fsf},\Acal^{\frakDp}_{\nsf})$,
	at least one of the following four conditions holds.
	\begin{myenumerate}
		\item\label{item6_diag_LMPautomata} 
			There exists a path
			\begin{align}\label{eqn7_diag_LMPautomata}
				q_0'\xrightarrow[]{s_1'}q_1'\xrightarrow[]{e_{2}'}q_2'\xrightarrow[]{s_{3}'}q_{3}',
			\end{align}
			where $q_0'\in Q_0'$, $s_1',s_3'\in(E_o')^*$, $e_2'\in E_o'$, $q_1'\xrightarrow[]
			{e_{2}'}q_{2}'$ is faulty, and $q_{3}'$ belongs to a dominant transition cycle
			consisting of observable transitions. 
		\item\label{item7_diag_LMPautomata}
			There exists a path
			\begin{align}\label{eqn8_diag_LMPautomata}
				q_0'\xrightarrow[]{s_1'}q_1'\xrightarrow[]{e_{2}'}q_2'\xrightarrow[]{s_{3}'}q_{3}',
			\end{align}
			where $q_0'\in Q_0'$, $s_1',s_3'\in(E_o')^*$, $e_2'\in E_o'$, $q_1'\xrightarrow[]
			{e_{2}'}q_{2}'$ is faulty, $q_{3}'(L)$ is eventually dominant
			in $\Acal^{\frakDp}_{\fsf}$, and $q_3'(R)$ is eventually dominant in $\Acal^{\frakDp}_{\nsf}$.
		\item\label{item8_diag_LMPautomata}
			There exists a path
			\begin{align}\label{eqn9_diag_LMPautomata}
				q_0'\xrightarrow[]{s_1'}q_1'\xrightarrow[]{s_{2}'}q_2'\xrightarrow[]{e_{3}'}q_{3}',
			\end{align}
			where $q_0'\in Q_0'$, $s_1'\in(E_o')^*$, $s_2'\in(E_{uo}')^*$, $e_3'\in E_f\times\{\ep\}$, 
			$q_{3}'(L)$ is eventually dominant
			in $\Acal^{\frakDp}_{\fsf}$, and $q_3'(R)$ is eventually dominant in $\Acal^{\frakDp}_{\nsf}$.
		\item\label{item27_diag_LMPautomata}  
			There exists a path
			\begin{align}\label{eqn22_diag_LMPautomata} 
				q_0'\xrightarrow[]{s_1'}q_1'\xrightarrow[]{s_{2}'}q_2'\xrightarrow[]{e_{3}'}q_{3}',
			\end{align}
			where $q_0'\in Q_0'$, $s_1'\in(E_o')^*$, $s_2'\in(E_{uo}')^*$, $e_3'\in E_f\times\{\ep\}$, 
			$q_{3}'(L)$ is eventually dominant
			in $\Acal^{\frakDp}_{\fsf}$, and some state in $q_1'(L)\xrightarrow[]{s_{2}'(L)}q_2'(L)$ is anti-dominant
			in $\Acal_{\fsf}^{\frakDp}$.
	\end{myenumerate}
\end{theorem}

\begin{proof}
	The proof of this result is put in Appendix~\ref{sec:proofThm3.12}.
\end{proof}

\begin{corollary}\label{cor1_diag_LMPautomata}
	A labeled ${\undernonnegQ}$-automaton $\Acal^{\undernonnegQ}$ is not $E_f$-diagnosable if and only if in the
	corresponding concurrent composition $\CCa(\Acal^{\undernonnegQ}_{\fsf},\Acal^{\undernonnegQ}_{\nsf})$,
	at least one of \eqref{item6_diag_LMPautomata}, \eqref{item7_diag_LMPautomata}, and 
	\eqref{item8_diag_LMPautomata} adapted to $\Acal^{\undernonnegQ}$ holds (\eqref{item27_diag_LMPautomata} 
	must not hold because $\Acal^{\undernonnegQ}$ contains no transition with negative weight).
\end{corollary}

\begin{example}\label{exam4_diag_LMPautomata}
	Reconsider the automaton $\Acal_1^{\underN}$ in Example~\ref{exam2_diag_LMPautomata} (shown in 
	Fig.~\ref{fig3_diag_LMPautomata}). In its concurrent composition $\CCa(\Acal_{1\fsf}^{\underN},
	\Acal_{1\nsf}^{\underN})$ (in Fig.~\ref{fig5_diag_LMPautomata}), the observable transition $(q_0,q_0)
	\xrightarrow[]{(a,a)}(q_3,q_4)$ is faulty, because it has two admissible paths $q_0\xrightarrow[]{\cyan f}
	q_1\xrightarrow[]{a}q_3$ and $q_0\left( \xrightarrow[]{u}q_2 \right)^3\xrightarrow[]{a}q_4$ both 
	with weight $4$; moreover, the observable self-loop $(q_3,q_4)\xrightarrow[]{(a,a)}(q_3,q_4)$
	is dominant, because it has two admissible paths $q_3\xrightarrow[]{a}q_3$ and $q_4\xrightarrow[]{a}q_4$
    both with positive weight $1$. Hence item~\eqref{item6_diag_LMPautomata} of Theorem~\ref{thm1_diag_LMPautomata}
	is satisfied, $\Acal_1^{\underN}$ is not $\cyan\{f\}$-diagnosable, which is consistent with the result 
	obtained in Example~\ref{exam3_diag_LMPautomata}. Item~\eqref{item7_diag_LMPautomata} is not satisfied.
	Item~\eqref{item8_diag_LMPautomata} is satisfied, because in path $(q_0,q_0)\xrightarrow[]{({\cyan f},\ep)}
	(q_1,q_0)$, $q_1$ is dominant in $\Acal_{1\fsf}^{\underN}$ (because of the unobservable self-loop on
	$q_1$ with positive weight) and $q_0$ is eventually dominant in $\Acal_{1\nsf}^{\underN}$ (because of 
	the unobservable path $q_0\xrightarrow[]{u}q_2\xrightarrow[]{u}q_2$). \eqref{item27_diag_LMPautomata}
	is not satisfied. 
\end{example}

Next we give an example for which only \eqref{item27_diag_LMPautomata} of Theorem~\ref{thm1_diag_LMPautomata}
is satisfied.

\begin{example}\label{exam5_diag_LMPautomata}
	Consider a second stuck-free labeled $\underQ$-automaton $\Acal_2^{\underQ}$ shown in 
	Fig.~\ref{fig6_diag_LMPautomata}. One easily sees that in $\CCa(\Acal_{2\fsf}^{\underQ},\Acal_{2\nsf}^{\underQ})$,
	there is no observable transition. Then neither \eqref{item6_diag_LMPautomata} nor \eqref{item7_diag_LMPautomata}
	is satisfied. \eqref{item8_diag_LMPautomata} is not satisfied either, because in $\Acal_{2\nsf}^{\underQ}$
	there is no reachable eventually dominant state. \eqref{item27_diag_LMPautomata} is satisfied, because in
	path $q_0\xrightarrow[]{u}q_1 \xrightarrow[]{u}q_1 \xrightarrow[]{{\cyan f}}q_2$, $q_1 \xrightarrow[]{u}q_1$
	is an anti-dominant cycle, $q_2$ is dominant. Hence $\Acal_2^{\underQ}$ is not $\cyan \{f\}$-diagnosable.
	Moreover, choose paths 
	\begin{align*}
		q_0\xrightarrow[]{u}q_1\left( \xrightarrow[]{u}q_1 \right)^{t+2} \xrightarrow[]{{\cyan f}}q_2 =: {\red\pi_t},\\
		q_2\left( \xrightarrow[]{u}q_2 \right)^{t+1} =: {\green\pi_t'},\\
		q_0\xrightarrow[]{u}q_1 =: {\blue\pi_t''},
	\end{align*}
	where $t\in \Z_+$. Then one has
	\begin{align*}
		&\ell(\tau({\red\pi_t}{\green\pi_t'})) = \ell(\tau({\blue\pi_t''})) = \ep,\\
		&\WEI_{{\green\pi_t'}} = t+1 > t,\\
		&\WEI_{{\blue\pi_t''}} = 0 > \WEI_{{\red\pi_t}} + t = -2.
	\end{align*}
	By Proposition~\ref{prop2_diag_LMPautomata}, one also has $\Acal_2^{\underQ}$ is not $\cyan \{f\}$-diagnosable.
	
	\begin{figure}[!htbp]
	\centering
	\begin{tikzpicture}
	[>=stealth',shorten >=1pt,thick,auto, node distance=2.5 cm, scale = 0.8, transform shape,
	->,>=stealth,inner sep=2pt, initial text = 0]

	\tikzstyle{emptynode}=[inner sep=0,outer sep=0]

	\node[initial, state, initial where = above] (q0) {$q_0$};
	\node[state] (q1) [right of = q0] {$q_1$};
	\node[state] (q2) [right of = q1] {$q_2$};
	\node[state] (q3) [right of = q2] {$q_3$};

	\path [->]
	(q0) edge node [above, sloped] {$u/0$} (q1)
	(q1) edge [loop above] node [above, sloped] {$u/-1$} (q1)
	(q1) edge node [above, sloped] {${\cyan f}/0$} (q2)
	(q2) edge [loop above] node [above, sloped] {$u/1$} (q2)
	(q2) edge node [above, sloped] {$u/0$} (q3)
	(q3) edge [loop above] node [above, sloped] {$a/1$} (q3)
	;

     \end{tikzpicture}
	 \caption{A labeled $\underQ$-automaton $\Acal_2^{\underQ}$, where $\ell(a)=a$,
	 $\ell({\cyan f})=\ell(u)=\ep$, only $\cyan f$ is faulty.}
	 \label{fig6_diag_LMPautomata} 
	 \end{figure}
\end{example}

\subsection{The complexity of verifying diagnosability of labeled}

We give the following complexity result on verifying $E_f$-diagnosability of 
$\Acal^{\underline{\mathbb{Q}}}$.

\begin{theorem}\label{thm3_diag_LMPautomata}
	The problem of verifying $E_f$-diagnosability of a labeled $\underQ$-automaton $\Acal^{\underQ}$ 
	is $\coNP$-complete, where $\coNP$-hardness even holds for deterministic, deadlock-free, and
	divergence-free $\Acal^{\underN}$.
\end{theorem}

\begin{proof}
	``$\coNP$-membership'':

	Recall that the concurrent composition $\CCa(\Acal^{\underQ}_{\fsf},\Acal^{\underQ}_{\nsf})$
	can be computed in $\NP$ (Theorem~\ref{thm4_diag_LMPautomata}).

	In \eqref{item6_diag_LMPautomata}, by definition, we can assume $q_3'$ belongs to a dominant simple 
	transition cycle without loss of generality. By Proposition~\ref{prop3_diag_LMPautomata}, in 
	$\CCa(\Acal^{\underQ}_{\fsf},\Acal^{\underQ}_{\nsf})$, whether an observable transition is faulty can be 
	verified in $\NP$, whether a simple observable transition cycle is dominant can also be verified in 
	$\NP$, so \eqref{item6_diag_LMPautomata} can be verified in $\NP$.

	By Proposition~\ref{prop4_diag_LMPautomata}, whether a state of $\Acal^{\underQ}$ is dominant (resp.,
	anti-dominant) can be verified in polynomial time. Hence in \eqref{item7_diag_LMPautomata},
	\eqref{item8_diag_LMPautomata}, and \eqref{item27_diag_LMPautomata},
	whether $q_{3}'(L)$ is eventually dominant in $\Acal^{\Q}_{\fsf}$, whether $q_3'(R)$ is eventually
	dominant in $\Acal^{\Q}_{\nsf}$, and whether some state in $q_1'(L)
	\xrightarrow[]{s_{2}'(L)}q_2'(L)$ is anti-dominant in $\Acal_{\fsf}^{\underQ}$  can be verified 
	in polynomial time. Hence \eqref{item7_diag_LMPautomata},
	\eqref{item8_diag_LMPautomata}, and \eqref{item27_diag_LMPautomata} can be verified in $\NP$.

	``$\coNP$-hardness'': 

	This proof is based on a polynomial-time reduction from the $\NP$-complete subset sum problem 
	\cite{Garey1990ComputerIntractability} to the negation of diagnosability, 
	and the reduction is the same as the ``coNP-hardness'' proof in 
    \cite[Theorem~3.7]{Zhang2021DiagWeightedAutoMonoidCDC}. We omit the specific procedure.
\end{proof}

  As a byproduct, by using the above reduction from the subset sum problem, we can prove 
  that the problem of computing concurrent composition $\CCa(\Acal^{\underN}_{\fsf},\Acal^{\underN}_{\nsf})$
  is $\NP$-hard in the size of $\Acal^{\underN}$. This result holds because the problem of computing 
  $\CCa(\Acal^{\underN}_{\fsf},\Acal^{\underN}_{\nsf})$ is a yes/no problem (i.e., whether there is an observable
  transition between two given states in $Q\times Q$ under a given observable event in $E_o'$ as in 
  Definition~\ref{def_CC_diag_LMPautomata}), and the subset sum problem instance $n_1,\dots,n_m,N$
  has a solution if and only if there is an observable
  transition $(q_0,q_0)\xrightarrow[]{(a,a)}(q_{m+2}^1,q_{m+2}^2)$ in the automaton constructed from 
  $n_1,\dots,n_m,N$ (see \cite[Fig.~1]{Zhang2021DiagWeightedAutoMonoidCDC}).

\begin{corollary}\label{cor2_diag_LMPautomata}
  It is $\NP$-hard to compute concurrent composition $\CCa(\Acal^{\underN}_{\fsf},\Acal^{\underN}_{\nsf})$ for a 
	deterministic, deadlock-free, and divergence-free automaton $\Acal^{\underN}$.
\end{corollary}

\subsection{A short discussion on diagnosability of labeled
\texorpdfstring{$\underline{\mathbb{Q}^k}$}{Qk}-automata
\texorpdfstring{$\Acal^{\underline{\mathbb{Q}^k}}$}{AQk}
over the vector progressive dioid \texorpdfstring{$\underline{\Q^k}$}{Qk}}
\label{sec:diagextension}

	Consider a labeled weighted automaton $\Acal^{\underline{\Q^k}}$ over vector progressive dioid
	$\underline{\Q^k}$~\eqref{eqn24_diag_LMPautomata}. We have shown in Example~\ref{exam6_diag_LMPautomata} that
	the eventually dominant (anti-dominant) elements are exactly the vectors whose first components are positive
	(resp., negative). Then by definition, we have a state $q$ is dominant
	(resp., anti-dominant) if and only if there is an unobservable transition cycle $q\xrightarrow[]{s}q=:\pi$
	such that the first component of $\WEI_{\pi}$ is positive (resp., negative). Then similar to
	Proposition~\ref{prop4_diag_LMPautomata}, whether a state $q$ is dominant (anti-dominant) can be checked in 
	time polynomial in the size of $\Acal^{\underline{\Q^k}}$. On the other hand, similar to
	Theorem~\ref{thm4_diag_LMPautomata}, we have concurrent composition $\CCao(\Acal^{\underline{\Q^k}}_{\fsf},
	\Acal^{\underline{\Q^k}}_{\nsf})$ can also be computed in $\NP$ in the size of $\Acal^{\underline{\Q^k}}$.
	Furthermore,
	we have in $\CCao(\Acal^{\underline{\Q^k}}_{\fsf},\Acal^{\underline{\Q^k}}_{\nsf})$, 
	a run $q_0'\xrightarrow[]{e_1'}\cdots\xrightarrow[]{e_n'}q_n'$ is dominant if and only if
	each transition $q_i'\xrightarrow[]{e_{i+1}'}q_{i+1}'$ has an admissible path
	$\pi_i$ such that the first component of $\sum_{i=0}^{n-1}\WEI_{\pi_i}$ is positive. Then similar
	to Proposition~\ref{prop3_diag_LMPautomata}, we also have in $\CCao(\Acal^{\underline{\Q^k}}_{\fsf},
	\Acal^{\underline{\Q^k}}_{\nsf})$, whether a simple cycle is dominant can be checked in $\NP$ in the size of
	$\Acal^{\underline{\Q^k}}$.
	Then for automaton $\Acal^{\underline{\Q^k}}$, the conditions in Theorem~\ref{thm1_diag_LMPautomata} can be
	checked in $\NP$. That is, diagnosability of $\Acal^{\underline{\Q^k}}$ can be verified in $\coNP$. 
	These results are summarized as follows.

\begin{theorem}\label{thm5_diag_LMPautomata}
  The problem of verifying $E_f$-diagnosability of a labeled $\underline{\Q^k}$-automaton $\Acal^{\underline{\Q^k}}$ 
  is $\coNP$-complete. The concurrent composition $\CCa(\Acal^{\underline{\Q^k}}_{\fsf},
  \Acal^{\underline{\Q^k}}_{\nsf})$ can be computed in $\NP$. 
\end{theorem}

\begin{example}\label{exam7_diag_LMPautomata}
  We revisit Example~\ref{exam5_diag_LMPautomata} (see Fig.~\ref{fig6_diag_LMPautomata}). The labeled 
  $\underQ$-automaton $\Acal_2^{\underQ}$ shown in Fig.~\ref{fig6_diag_LMPautomata} can be regarded as an autonomous 
  robot working on the straight line $\R^1$, where the ``autonomous'' comes from the fact that the occurrences
  of events (hence the dynamics of the automaton) are spontaneous. Every transition is endowed with a possible
  position deviation represented by the weight, e.g., transition
  $q_0\xrightarrow[]{u/0}q_1$ means along with the transition, the robot stays in the same position (with
  position deviation $0$), while transition $q_1\xrightarrow[]{u/-1}q_1$ means the robot moves $1$ unit of
  distance to the left, weight $1$ means moving $1$ unit of distance to the right.

  We add a time dimension into the automaton to measure the time consumptions for the transitions as in 
  Fig.~\ref{fig9_diag_LMPautomata}. 
  \begin{figure}[!htbp]
	\centering
	\begin{tikzpicture}
	[>=stealth',shorten >=1pt,thick,auto, node distance=2.5 cm, scale = 0.8, transform shape,
	->,>=stealth,inner sep=2pt, initial text = {$(0,0)$}]

	\tikzstyle{emptynode}=[inner sep=0,outer sep=0]

	\node[initial, state, initial where = above] (q0) {$q_0$};
	\node[state] (q1) [right of = q0] {$q_1$};
	\node[state] (q2) [right of = q1] {$q_2$};
	\node[state] (q3) [right of = q2] {$q_3$};

	\path [->]
	(q0) edge node [above, sloped] {$u/(1,0)$} (q1)
	(q1) edge [loop above] node [above, sloped] {$u/(1,-1)$} (q1)
	(q1) edge node [above, sloped] {${\cyan f}/(1,0)$} (q2)
	(q2) edge [loop above] node [above, sloped] {$u/(1,1)$} (q2)
	(q2) edge node [above, sloped] {$u/(1,0)$} (q3)
	(q3) edge [loop above] node [above, sloped] {$a/(1,1)$} (q3)
	;

     \end{tikzpicture}
	 \caption{A labeled $\underline{\Q^2}$-automaton $\Acal_3^{\underline{\Q^2}}$, where dioid 
	 $\underline{\Q^2}$ is shown in \eqref{eqn24_diag_LMPautomata}, 
	 $\ell(a)=a$, $\ell({\cyan f})=\ell(u)=\ep$, only $\cyan f$ is faulty.}
	 \label{fig9_diag_LMPautomata}  
	 \end{figure}

	 In the weights of the transitions of automaton $\Acal_3^{\underline{\Q^2}}$, the second components denote 
	 the position deviations along with the corresponding transitions (which is the same as those in automaton
	 $\Acal_2^{\underQ}$), while the first components denote the time consumptions for the transitions' 
	 executions. We assume for simplicity that the time consumptions for all transitions are $1$ unit of time.
	 When the robot is working, we first consider time consumptions and then position deviations, i.e.,
	 the time dimension takes precedence over the space dimension, so in 
	 each vector of $\underline{\Q^2}$, the time component is before the space component, which is consistent
	 with the lexicographic order of the progressive dioid $\underline{\Q^2}$ as in \eqref{eqn24_diag_LMPautomata}.

     Next, we use Theorem~\ref{thm1_diag_LMPautomata} to verify the $\{{\cyan f}\}$-diagnosability of 
	 $\Acal_3^{\underline{\Q^2}}$.
	 In $\CCa(\Acal^{\underline{\Q^2}}_{3\fsf},\Acal^{\underline{\Q^2}}_{3\nsf})$, there is no observable transition,
	 because in $\Acal^{\underline{\Q^2}}_{3\nsf}$ there is no observable transition. Hence neither 
	 \eqref{item6_diag_LMPautomata} nor \eqref{item7_diag_LMPautomata} in Theorem~\ref{thm1_diag_LMPautomata} 
	 is satisfied. In $\Acal_3^{\underline{\Q^2}}$, there is no weight that is eventually anti-dominant (i.e.,
	 those whose first components are negative as shown in Example~\ref{exam6_diag_LMPautomata}),  
	 so \eqref{item27_diag_LMPautomata} in Theorem~\ref{thm1_diag_LMPautomata} is not satisfied either.
	 Part of $\CCa(\Acal^{\underline{\Q^2}}_{3\fsf},\Acal^{\underline{\Q^2}}_{3\nsf})$ is depicted in 
	 Fig.~\ref{fig8_diag_LMPautomata}.
	 \begin{figure}[!htbp]
	\centering
	\begin{tikzpicture}
	[>=stealth',shorten >=1pt,thick,auto, node distance=3.5 cm, scale = 0.8, transform shape,
	->,>=stealth,inner sep=2pt]

	\tikzstyle{emptynode}=[inner sep=0,outer sep=0]

	\node[initial, state, initial where = above] (q0q0) {$q_0q_0$};
	\node[state] (q1q0) [right of = q0q0] {$q_1q_0$};
	\node[state] (q2q0) [right of = q1q0] {$q_2q_0$};
	\node[state] (q2q1) [right of = q2q0] {$q_2q_1$};

	\path [->]
	(q0q0) edge node [above, sloped] {$(u,\ep)/(1,0)$} (q1q0)
	(q1q0) edge [loop above] node [above, sloped] {$(u,\ep)/(1,-1)$} (q1q0)
    (q1q0) edge node [above, sloped] {$({\cyan f},\ep)/(1,0)$} (q2q0)
	(q2q0) edge [loop above] node [above, sloped] {$(u,\ep)/(1,1)$} (q2q0)
	(q2q0) edge node [above, sloped] {$(\ep,u)/(-1,0)$} (q2q1)
	(q2q1) edge [loop above] node [above, sloped] {$(\ep,u)/(-1,1)$} (q2q1)
	;

     \end{tikzpicture}
	 \caption{Part of concurrent composition $\CCa(\Acal^{\underline{\Q^2}}_{3\fsf},\Acal^{\underline{\Q^2}}_{3\nsf})$,
	 where $\Acal^{\underline{\Q^2}}_{3}$ is shown in Fig.~\ref{fig9_diag_LMPautomata}.}
	 \label{fig8_diag_LMPautomata}  
	 \end{figure}

	 By Fig.~\ref{fig8_diag_LMPautomata}, \eqref{item8_diag_LMPautomata} in Theorem~\ref{thm1_diag_LMPautomata} 
	 is satisfied, because event $({\cyan f},\ep)$ is faulty, $q_2$ is eventually dominant in $\Acal_{3\fsf}^{
	 \underline{\Q^2}}$, $q_1$ is eventually dominant in $\Acal_{3\nsf}^{\underline{\Q^2}}$. Therefore, 
	 automaton $\Acal_3^{\underline{\Q^2}}$ is not $\{{\cyan f}\}$-diagnosable. 

	 Comparing automaton $\Acal_3^{\underline{\Q^2}}$ (Fig.~\ref{fig9_diag_LMPautomata}) with automaton 
	 $\Acal_2^{\underQ}$ (Fig.~\ref{fig6_diag_LMPautomata}), the same thing is that neither of them is 
	 $\{{\cyan f}\}$-diagnosable, but the reasons resulting in their non-diagnosability are different:
	 $\Acal_3^{\underline{\Q^2}}$ satisfies only the third condition of Theorem~\ref{thm1_diag_LMPautomata},
	 while $\Acal_2^{\underQ}$ satisfies only the fourth condition of Theorem~\ref{thm1_diag_LMPautomata}.

	 Fig.~\ref{fig9_diag_LMPautomata} shows a simple scenario for which an autonomous robot is working on the
	 straight line
	 $\R^1$. This example can be naturally extended to scenarios of multi-dimensional spaces $\R^2$ and $\R^3$, or
	 even $\R^k$. Moreover, according to the precedence of considering the time dimension and different space 
	 dimensions, one can adjust the order of different components in the corresponding vectors of progressive dioids 
	 to be consistent with the lexicographic orders therein. On the other hand, practical meanings of states
	 are not specified here. For example, the meanings of states are specified as pairs of energy levels and positions
	 of a robot in Example~1 and Example~5 of \cite{Zhang2022DetWAMonoid} for labeled weighted automata over monoids.
	 The readers who are interested in such scenarios
	 could try to adapt the scenarios to their applications.
\end{example}

\section{Conclusion}
\label{sec:conc}

In this paper, we formulated a notion of diagnosability for labeled weighted automata over a subclass 
of dioids called progressive, where the weight of a transition represents diverse physical meanings such time 
elapsing and position deviations. We gave a
notion of concurrent composition for such automata and used the notion 
to give a necessary and sufficient condition for diagnosability. We also proved that for the dioid
$(\Q\cup\{-\infty\},\max,+,-\infty,0)$ (which is progressive), the problem of computing 
a concurrent composition is
$\NP$-complete and verifying diagnosability is $\coNP$-complete, where $\NP$-hardness and
$\coNP$-hardness even hold for deterministic,
deadlock-free, and divergence-free automata. $\coNP$-membership has been extended 
to automaton $\Acal^{\underline{\Q^k}}$ over vector progressive dioid 
$\underline{\Q^k}$~\eqref{eqn24_diag_LMPautomata}. The $\coNP$ upper bound on verifying diagnosability
can be extended to labeled real-time automata (a subclass of labeled timed automata with a single clock)
by extending the current version of concurrent composition to labeled real-time automata.
Note that particularly when the weights of the transitions of the automata considered in the current paper are
nonnegative rational numbers, the automata become a subclass of labeled timed automata.
Recall that for labeled timed automata, the diagnosability verification problem is $\PSPACE$-complete 
(the $\PSPACE$ lower bound holds even for labeled 2-clock timed automata), hence in this paper a subclass
of labeled timed automata for which the diagnosability verification problem belongs to $\coNP$ was found.


\begin{thebibliography}{10}

\bibitem{Alur1994TimedAutomaton}
R.~Alur and D.~L. Dill.
\newblock A theory of timed automata.
\newblock {\em Theoretical Computer Science}, 126(2):183--235, 1994.

\bibitem{Atig2009YenPathLogicPetriNet}
M.~F. Atig and P.~Habermehl.
\newblock {On Yen's path logic for Petri nets}.
\newblock In Olivier Bournez and Igor Potapov, editors, {\em Reachability
  Problems}, pages 51--63, Berlin, Heidelberg, 2009. Springer Berlin
  Heidelberg.

\bibitem{Basile2009DiagnosisPetriNetsILP}
F.~{Basile}, P.~{Chiacchio}, and G.~{De Tommasi}.
\newblock An efficient approach for online diagnosis of discrete event systems.
\newblock {\em IEEE Transactions on Automatic Control}, 54(4):748--759, 2009.

\bibitem{Basilio2021ResilienceDES}
J.C. Basilio, C.N. Hadjicostis, and R.~Su.
\newblock Analysis and control for resilience of discrete event systems: Fault
  diagnosis, opacity and cyber security.
\newblock {\em Foundations and Trends\textsuperscript{ \textregistered}{ }in
  Systems and Control}, 8(4):285--443, 2021.

\bibitem{Berard2018DiagnosabilityPetriNet}
B.~B\'{e}rard, S.~Haar, S.~Schmitz, and S.~Schwoon.
\newblock {The complexity of diagnosability and opacity verification for Petri
  nets}.
\newblock {\em Fundamenta Informaticae}, 161(4):317--349, 2018.

\bibitem{Bouyer2005FaultDiagnosisTimedAutomata}
P.~Bouyer, F.~Chevalier, and D.~D'Souza.
\newblock Fault diagnosis using timed automata.
\newblock In {\em Proceedings of the 8th International Conference on
  Foundations of Software Science and Computation Structures}, FOSSACS'05,
  pages 219--233, Berlin, Heidelberg, 2005. Springer-Verlag.

\bibitem{Cabasino2005FaultDetectionPetriNets}
M.~P. Cabasino, A.~Giua, and C.~Seatzu.
\newblock Fault detection for discrete event systems using {P}etri nets with
  unobservable transitions.
\newblock {\em Automatica}, 46(9):1531--1539, 2010.

\bibitem{Cabasino2012DiagnosabilityPetriNet}
M.P. Cabasino, A.~Giua, S.~Lafortune, and C.~Seatzu.
\newblock {A new approach for diagnosability analysis of Petri nets using
  verifier nets}.
\newblock {\em IEEE Transactions on Automatic Control}, 57(12):3104--3117, Dec
  2012.

\bibitem{Cassez2012ComplexityCodiagnosability}
F.~{Cassez}.
\newblock The complexity of codiagnosability for discrete event and timed
  systems.
\newblock {\em IEEE Transactions on Automatic Control}, 57(7):1752--1764, July
  2012.

\bibitem{Cassez2008FaultDiagnosisStDyObser}
F.~Cassez and S.~Tripakis.
\newblock Fault diagnosis with static and dynamic observers.
\newblock {\em Fundamenta Informaticae}, 88(4):497--540, 2008.

\bibitem{Cieslak1988ObservabilityDES}
R.~Cieslak, C.~Desclaux, A.S. Fawaz, and P.~Varaiya.
\newblock Supervisory control of discrete-event processes with partial
  observations.
\newblock {\em IEEE Transactions on Automatic Control}, 33(3):249--260, 1988.

\bibitem{Cormen2009Algorithms}
T.H. Cormen, C.E. Leiserson, R.L. Rivest, and C.~Stein.
\newblock {\em Introduction to Algorithms}.
\newblock The MIT Press, 3rd edition, 2009.

\bibitem{Debouk2000CodiagnosabilityAutomata}
R.~Debouk, S.~Lafortune, and D.~Teneketzis.
\newblock Coordinated decentralized protocols for failure diagnosis of discrete
  event systems.
\newblock {\em Discrete Event Dynamic Systems}, 10(1):33--86, 2000.

\bibitem{Fearnley2015Reachability2ClocksTimeAutomataPSPACE}
J.~Fearnley and M.~Jurdzi\'{n}ski.
\newblock Reachability in two-clock timed automata is {PSPACE}-complete.
\newblock {\em Information and Computation}, 243:26--36, 2015.
\newblock 40th International Colloquium on Automata, Languages and Programming
  (ICALP 2013).

\bibitem{Garey1990ComputerIntractability}
M.R. Garey and D.S. Johnson.
\newblock {\em Computers and Intractability: A Guide to the Theory of
  NP-Completeness}.
\newblock W. H. Freeman \& Co., USA, 1990.

\bibitem{Giua2005FaultDetectionPetriNets}
A.~{Giua} and C.~{Seatzu}.
\newblock Fault detection for discrete event systems using {P}etri nets with
  unobservable transitions.
\newblock In {\em Proceedings of the 44th IEEE Conference on Decision and
  Control}, pages 6323--6328, 2005.

\bibitem{Hack1975PetriNetLanguage}
M.~Hack.
\newblock Petri net languages.
\newblock Technical report, Cambridge, MA, USA, 1975.

\bibitem{Jiang2001PolyAlgorithmDiagnosabilityDES}
S.~Jiang, Z.~Huang, V.~Chandra, and R.~Kumar.
\newblock A polynomial algorithm for testing diagnosability of discrete-event
  systems.
\newblock {\em IEEE Transactions on Automatic Control}, 46(8):1318--1321, Aug
  2001.

\bibitem{Keroglou2019AA-DiagnosabilityProbAutomata}
C.~{Keroglou} and C.~N. {Hadjicostis}.
\newblock Verification of {AA}-diagnosability in probabilistic finite automata
  is {PSPACE}-hard.
\newblock In {\em 2019 IEEE 58th Conference on Decision and Control (CDC)},
  pages 6712--6717, 2019.

\bibitem{Lai2022DiagnosabilityMPAutomata}
A.~Lai, J.~Komenda, and S.~Lahaye.
\newblock Diagnosability of unambiguous max-plus automata.
\newblock {\em IEEE Transactions on Systems, Man, and Cybernetics: Systems},
  52(11):7302--7311, 2022.

\bibitem{Lin1994DiagnosabilityDES}
F.~Lin.
\newblock Diagnosability of discrete event systems and its applications.
\newblock {\em Discrete Event Dynamic Systems}, 4(2):197--212, May 1994.

\bibitem{Moreira2011Codiagnosability}
M.~V. {Moreira}, T.~C. {Jesus}, and J.~C. {Basilio}.
\newblock Polynomial time verification of decentralized diagnosability of
  discrete event systems.
\newblock {\em IEEE Transactions on Automatic Control}, 56(7):1679--1684, July
  2011.

\bibitem{Nykanen2002ExactPathLength}
M.~Nyk\"{a}nen and E.~Ukkonen.
\newblock The exact path length problem.
\newblock {\em Journal of Algorithms}, 42(1):41--53, 2002.

\bibitem{Qiu2006DecentralizedFD}
W.~Qiu and R.~Kumar.
\newblock Decentralized failure diagnosis of discrete event systems.
\newblock {\em IEEE Transactions on Systems, Man, and Cybernetics - Part A:
  Systems and Humans}, 36:384--395, 2006.

\bibitem{RabinScott1959PowersetConstruction}
M.O. Rabin and D.~Scott.
\newblock Finite automata and their decision problems.
\newblock {\em IBM Journal of Research and Development}, 3(2):114--125, 1959.

\bibitem{Sampath1995DiagnosabilityDES}
M.~Sampath, R.~Sengupta, S.~Lafortune, K.~Sinnamohideen, and D.~Teneketzis.
\newblock Diagnosability of discrete-event systems.
\newblock {\em IEEE Transactions on Automatic Control}, 40(9):1555--1575, Sep
  1995.

\bibitem{Shu2007Detectability_DES}
S.~Shu, F.~Lin, and H.~Ying.
\newblock Detectability of discrete event systems.
\newblock {\em IEEE Transactions on Automatic Control}, 52(12):2356--2359, Dec
  2007.

\bibitem{Thorsley2017EquivConditionDiagnosabilityStochDES}
D.~Thorsley.
\newblock A necessary and sufficient condition for diagnosability of stochastic
  discrete event systems.
\newblock {\em Discrete Event Dynamic Systems}, 27(3):481--500, 2017.

\bibitem{Thorsley2005DiagnosabilityStochDES}
D.~{Thorsley} and D.~{Teneketzis}.
\newblock Diagnosability of stochastic discrete-event systems.
\newblock {\em IEEE Transactions on Automatic Control}, 50(4):476--492, 2005.

\bibitem{Tripakis2002DiagnosisTimedAutomata}
S.~Tripakis.
\newblock Fault diagnosis for timed automata.
\newblock In Werner Damm and Ernst~R\"{u}diger Olderog, editors, {\em Formal
  Techniques in Real-Time and Fault-Tolerant Systems}, pages 205--221, Berlin,
  Heidelberg, 2002. Springer Berlin Heidelberg.

\bibitem{Viana2019CodiagnosabilityDES}
G.~S. Viana and J.~C. Basilio.
\newblock Codiagnosability of discrete event systems revisited: A new necessary
  and sufficient condition and its applications.
\newblock {\em Automatica}, 101:354--364, 2019.

\bibitem{Yen1992YenPathLogicPetriNet}
H.~C. Yen.
\newblock {A unified approach for deciding the existence of certain Petri net
  paths}.
\newblock {\em Information and Computation}, 96(1):119--137, 1992.

\bibitem{Yin2017DiagnosabilityLabeledPetriNets}
X.~Yin and S.~Lafortune.
\newblock {On the decidability and complexity of diagnosability for labeled
  Petri nets}.
\newblock {\em IEEE Transactions on Automatic Control}, 62(11):5931--5938, Nov
  2017.

\bibitem{Yoo2002DiagnosabiliyDESPTime}
T.-S. {Yoo} and S.~{Lafortune}.
\newblock Polynomial-time verification of diagnosability of partially observed
  discrete-event systems.
\newblock {\em IEEE Transactions on Automatic Control}, 47(9):1491--1495, Sep.
  2002.

\bibitem{Zhang2021UnifyingDetDiagPred}
K.~Zhang.
\newblock A unified method to decentralized state detection and fault
  diagnosis/prediction of discrete-event systems.
\newblock {\em Fundamenta Informaticae}, 181(4):339--371, 2021.

\bibitem{Zhang2022DetWAMonoid}
K.~Zhang.
\newblock Detectability of labeled weighted automata over monoids.
\newblock {\em Discrete Event Dynamic Systems}, 32(3):435--494, 2022.

\bibitem{Zhang2023DESbook}
K.~Zhang.
\newblock A new framework for discrete-event systems.
\newblock {\em Foundations and Trends\textsuperscript{\textregistered}{ }in
  Systems and Control}, 10(1-2):1--179, 2023.

\bibitem{Zhang2021DiagWeightedAutoMonoidCDC}
K.~Zhang and J.~Raisch.
\newblock Diagnosability of labeled weighted automata over the monoid
  $(\mathbb{Q}_{\ge0},+,0)$.
\newblock In {\em 2021 60th IEEE Conference on Decision and Control (CDC)},
  pages 6867--6872, 2021.

\end{thebibliography}

\appendix

\section{Proof of Theorem~\ref{thm1_diag_LMPautomata}}
\label{sec:proofThm3.12}

	``if'':

	Assume \eqref{item6_diag_LMPautomata} holds. Then in $\Acal^{\frakDp}$ there exist paths

	\begin{align*}
		&q_0'(L)\xrightarrow[]{\bar s_1} q_1'(L)\xrightarrow[]{\bar s_2} q_2'(L)\xrightarrow[]{\bar s_3}
		 q_3'(L)\xrightarrow[]{\bar s_4} q_3'(L),\\
		&q_0'(R)\xrightarrow[]{\hat s_1} q_1'(R)\xrightarrow[]{\hat s_2} q_2'(R)\xrightarrow[]{\hat s_3}
		 q_3'(R)\xrightarrow[]{\hat s_4} q_3'(R)
	\end{align*}
	such that
	\begin{align*}
		&q_0'(L),q_0'(R)\in Q_0,\\
		&\bar s_i\in E^*E_o, \hat s_i\in(E\setminus E_f)^*(E_o\setminus E_f),
		\ell(\bar s_i)=\ell(\hat s_i)\\
		&\text{for }i= 1,3,4,\\
		&\bar s_2\in(E_{uo})^*E_o, \hat s_2\in(E_{uo}\setminus E_f)^*(E_o\setminus E_f),\ell(\bar s_2)=\ell(\hat s_2),\\
		&\ell(\bar s_1)=\ell(s_1'), \ell(\bar s_2)=\ell(e_2'), \ell(\bar s_3)=\ell(s_3'),\\
		&\ell(\tau(q_{i-1}'(L)\xrightarrow[]{\bar s_{i}} q_{i}'(L)))=
		\ell(\tau(q_{i-1}'(R)\xrightarrow[]{\hat s_{i}} q_{i}'(R)))\\
		&\text{for all }i\in\llb 1,3\rrb,\\
		&\ell(\tau(q_{3}'(L)\xrightarrow[]{\bar s_{4}} q_{3}'(L)))=
		\ell(\tau(q_{3}'(R)\xrightarrow[]{\hat s_{4}} q_{3}'(R))),\\
		&E_f\in\bar s_2,\\
		&\!\WEI(q_{3}'(L)\xrightarrow[]{\bar s_{4}} q_{3}'(L))=\WEI(q_{3}'(R)\xrightarrow[]{\hat s_{4}} q_{3}'(R))=:l,\\
		&l\text{\ is eventually dominant}.
	\end{align*}
	Let $t\in T$ be such that $t$ is eventually dominant. Choose $n\in\N$ such that $l^n\succ t$ 
	and $\WEI(q_2'(L)\xrightarrow[]{\bar s_3}q_3'(L))\otimes l^n \succ{\bf1}$
	($n$ exists because $l$ is eventually dominant).

	Choose
	\begin{align*}
		q_0'(L)\xrightarrow[]{\bar s_1} q_1'(L)\xrightarrow[]{\bar s_2} q_2'(L)
		&=: \pi_t,\\
		q_2'(L)\xrightarrow[]{\bar s_3}q_3'(L)\left(\xrightarrow[]
		{\bar s_4} q_3'(L)\right)^{2n}
		&=: \pi_t',\\
		q_0'(R)\xrightarrow[]{\hat s_1} q_1'(R)\xrightarrow[]{\hat s_2} q_2'(R)\xrightarrow[]{\hat s_3}
		q_3'(R)\left(\xrightarrow[]{\hat s_4} q_3'(R)\right)^{2n}
		&=: \pi_t'',
	\end{align*}
	one then has $\ell(\tau(\pi_t\pi_t'))=\ell(\tau(\pi_t''))$, $E_f\in \bar s_2$,
	$E_f\notin \hat s_1\hat s_2\hat s_3(\hat s_4)^{2n}$, $\WEI_{\pi_t'} = \WEI(q_2'(L)\xrightarrow[]{\bar s_3}
	q_3'(L))\otimes l^{2n} 	\succeq l^n\succ t$, and
	$\WEI_{\pi_t''}=\WEI_{\pi_t}\otimes \WEI_{\pi_t'}\succeq \WEI_{\pi_t}\otimes t$.
	Then By Proposition~\ref{prop2_diag_LMPautomata}, $\Acal^{\frakDp}$ is not $E_f$-diagnosable.

	Assume \eqref{item7_diag_LMPautomata} holds. Then there exist unobservable paths 
	\begin{subequations}\label{eqn21_diag_LMPautomata} 
		\begin{align}
			&q_3'(L) \xrightarrow[]{s_4} q_4 \xrightarrow[]{s_5} q_4 \text{ in }\Acal^{\frakDp}_{\fsf},\\
			&q_3'(R) \xrightarrow[]{s_6} q_5 \xrightarrow[]{s_7} q_5 \text{ in }\Acal^{\frakDp}_{\nsf}
		\end{align}
	\end{subequations}
	such that $s_4\in(E_{uo})^*$, $s_5\in(E_{uo})^+$, $s_6\in(E_{uo}\setminus E_f)^*$, $s_7\in(E_{uo}\setminus E_f)^+$, 
	$\WEI(q_4 \xrightarrow[]{s_5} q_4)=:l_1$, $\WEI(q_5 \xrightarrow[]{s_7} q_5)=:l_2$, $l_1$ and $l_2$
	are both eventually dominant.

	By \eqref{eqn8_diag_LMPautomata}, in $\Acal^{\frakDp}$ there exist paths
	\begin{align*}
		&q_0'(L)\xrightarrow[]{\bar s_1} q_1'(L)\xrightarrow[]{\bar s_2} q_2'(L)
		\xrightarrow[]{\bar s_3}q_3'(L)
		\xrightarrow[]{s_4} q_4 \xrightarrow[]{s_5} q_4,\\
		&q_0'(R)\xrightarrow[]{\hat s_1} q_1'(R)\xrightarrow[]{\hat s_2} q_2'(R)\xrightarrow[]{\hat s_3}
		 q_3'(R)
		\xrightarrow[]{s_6} q_5 \xrightarrow[]{s_7} q_5
	\end{align*}
	such that
	\begin{align*}
		&q_0'(L),q_0'(R)\in Q_0,\\
		&\bar s_i\in E^*E_o, \hat s_i\in(E\setminus E_f)^*(E_o\setminus E_f)\text{ for }i=1,3,\\
		&\bar s_2\in (E_{uo})^*E_o, \hat s_2\in (E_{uo}\setminus E_f)^*(E_o\setminus E_f),\\
		&\ell(\bar s_i)=\ell(\hat s_i)=\ell(s_i')\text{ for }i=1,3,\\
		&\ell(\bar s_2)=\ell(\hat s_2)=\ell(e_2'),\\
		&\ell(\tau(q_{i-1}'(L)\xrightarrow[]{\bar s_{i}} q_{i}'(L)))=\ell(\tau(q_{i-1}'(R)\xrightarrow[]{\hat s_{i}} q_{i}'(R)))\\
		&\text{for }i=1,2,3,\\
		&E_f\in\bar s_2.
	\end{align*}

	Let $t\in T$ be eventually dominant. Choose $n\in\N$ such that $(l_1)^n\succ t$, $(l_2)^n\succ t$,
	$\WEI(q_2'(L)\xrightarrow[]{\bar s_3}q_3'(L) \xrightarrow[]{s_4}q_4)\otimes (l_1)^n\succ {\bf1}$, and
	$\WEI(q_2'(R)\xrightarrow[]{\hat s_3}q_3'(R) \xrightarrow[]{s_6}q_5)\otimes (l_2)^n\succ{\bf1}$.

	Denote
	\begin{align*}
		q_0'(L)\xrightarrow[]{\bar s_1} q_1'(L)\xrightarrow[]{\bar s_2} q_2'(L)
		&=: \pi_t,\\
		q_2'(L)\xrightarrow[]{\bar s_3}q_3'(L) \xrightarrow[]{s_4}q_4 \left(\xrightarrow[]
		{s_5}q_4\right)^{2n}
		&=: \pi_t',\\
		q_0'(R)\xrightarrow[]{\hat s_1} q_1'(R)\xrightarrow[]{\hat s_2} q_2'(R) &\\
		\xrightarrow[]{\hat s_3} q_3'(R) \xrightarrow[]{s_6}q_5
		\left(\xrightarrow[]{s_7} q_5\right)^{2n}
		&=: \pi_t'',
	\end{align*}
	One then has $E_f\in\bar s_2$, $E_f\notin\hat s_1\hat s_2\hat s_3s_6
	(s_7)^{2n}$, $\ell(\tau(\pi_t\pi_t'))=\ell(\tau(\pi_t''))$,
	$\WEI_{\pi_t'} = \WEI(q_2'(L)\xrightarrow[]{\bar s_3}q_3'(L) \xrightarrow[]{s_4}q_4)\otimes (l_1)^{2n}
	\succeq (l_1)^n\succ t$, and $\WEI_{\pi_{t}''} = \WEI(q_0'(R)\xrightarrow[]{\hat s_1} q_1'(R)\xrightarrow[]
	{\hat s_2} q_2'(R))\otimes \WEI(q_2'(R)\xrightarrow[]
	{\hat s_3}q_3'(R) \xrightarrow[]{s_6}q_5)\otimes (l_2)^{2n} = \WEI_{\pi_t}\otimes \WEI(q_2'(R)\xrightarrow[]
	{\hat s_3}q_3'(R) \xrightarrow[]{s_6}q_5)\otimes (l_2)^{2n} \succeq \WEI_{\pi_t}\otimes (l_2)^n 
	\succeq\WEI_{\pi_t}\otimes t$.

	Then also by Proposition~\ref{prop2_diag_LMPautomata}, $\Acal^{\frakDp}$ is not $E_f$-diagnosable.

	Assume \eqref{item8_diag_LMPautomata} holds. Then similar to \eqref{item7_diag_LMPautomata}, 
	there also exist unobservable paths in $\Acal^{\frakDp}$ as in \eqref{eqn21_diag_LMPautomata}.

	By \eqref{eqn9_diag_LMPautomata}, in $\Acal^{\frakDp}$ there exist paths
	\begin{align*}
		&q_0'(L)\xrightarrow[]{\bar s_1} q_1'(L)\xrightarrow[]{s_2'(L)} q_2'(L)
		\xrightarrow[]{e_f}q_3'(L)\xrightarrow[]{s_4} q_4 \xrightarrow[]{s_5} q_4, \\
		&q_0'(R)\xrightarrow[]{\hat s_1} q_1'(R)\xrightarrow[]{s_2'(R)} q_2'(R)\xrightarrow[]{\ep}
		 q_3'(R) \xrightarrow[]{s_6} q_5 \xrightarrow[]{s_7} q_5\\
	\end{align*}
	such that
	\begin{align*}
		&q_0'(L),q_0'(R)\in Q_0,\\
		&\bar s_1\in E^*E_o, \hat s_1\in(E\setminus E_f)^*(E_o\setminus E_f),\\
		&\ell(\bar s_1)=\ell(\hat s_1)=\ell(s_1'),\\
		&\ell(\tau(q_{0}'(L)\xrightarrow[]{\bar s_{1}} q_{1}'(L)))=
		\ell(\tau(q_{0}'(R)\xrightarrow[]{\hat s_{1}} q_{1}'(R))),\\
		&\WEI(q_{0}'(L)\xrightarrow[]{\bar s_{1}} q_{1}'(L))=
		\WEI(q_{0}'(R)\xrightarrow[]{\hat s_{1}} q_{1}'(R)),\\
		&q_2'(R)=q_3'(R),\quad (e_f,\ep)= e_3'.
	\end{align*}

	Let $t\in T$ be eventually dominant. Choose a sufficiently large $n\in\N$ such that
	$(l_1)^n\succ t$, $(l_2)^n\succ t$, $\WEI(q_3'(L) \xrightarrow[]{s_4}q_4)\otimes (l_1)^n\succ {\bf1}$,
	$\WEI(q_1'(R)\xrightarrow[]{s_2'(R)} q_2'(R)\xrightarrow[]{s_6}q_5 )\otimes (l_2)^n\succ{\bf1}$, and
	$(l_2)^n\succ \WEI(q_1'(L)\xrightarrow[]{s_2'(L)} q_2'(L)\xrightarrow[]{e_f}q_3'(L))$.

	Denote
	\begin{align*}
		q_0'(L)\xrightarrow[]{\bar s_1} q_1'(L)\xrightarrow[]{s_2'(L)} q_2'(L)\xrightarrow[]
		{e_f}q_3'(L) &=: \pi_t,\\
		q_3'(L) \xrightarrow[]{s_4}q_4 \left(\xrightarrow[]
		{s_5}q_4 \right)^{2n}
		&=: \pi_t',\\
		q_0'(R)\xrightarrow[]{\hat s_1} q_1'(R)\xrightarrow[]{s_2'(R)} q_2'(R)
		 \xrightarrow[]{s_6}q_5\left(\xrightarrow[]{s_7} q_5\right)^{3n}
		&=: \pi_t''.
	\end{align*}
	One then has $E_f\notin\hat s_1s_2'(R)s_6(s_7)^{3n}$, $\ell(\tau(\pi_t\pi_t'))=\ell(\tau(\pi_t''))=
	\ell(\tau(q_0'(L)\xrightarrow[]{\bar s_1} q_1'(L)))$,
	$\WEI_{\pi_t'} = \WEI(q_3'(L) \xrightarrow[]{s_4}q_4)\otimes (l_1)^{2n}
	\succeq (l_1)^n \succ t$, and $\WEI_{\pi_{t}''}= \WEI(q_0'(R)\xrightarrow[]{\hat s_1} q_1'(R))
	\otimes \WEI(q_1'(R)\xrightarrow[]{s_2'(R)}q_2'(R)
	\xrightarrow[]{s_6}q_5)\otimes (l_2)^{3n} = \WEI(q_0'(L)\xrightarrow[]{\bar s_1} q_1'(L)) \otimes
	\WEI(q_1'(R)\xrightarrow[]{s_2'(R)}q_2'(R) \xrightarrow[]{s_6}q_5)\otimes (l_2)^{3n}
	\succeq \WEI(q_0'(L)\xrightarrow[]{\bar s_1} q_1'(L))\otimes (l_2)^{2n} \succeq
	\WEI(q_0'(L)\xrightarrow[]{\bar s_1} q_1'(L))\otimes
	\WEI(q_1'(L)\xrightarrow[]{s_2'(L)} q_2'(L)\xrightarrow[]{e_f}q_3'(L)) \otimes (l_2)^n 
	\succeq \WEI_{\pi_t}\otimes t $.

	Also by Proposition~\ref{prop2_diag_LMPautomata}, $\Acal^{\frakDp}$ is not $E_f$-diagnosable.

	Assume \eqref{item27_diag_LMPautomata} holds. Then there exist unobservable paths 
	\begin{subequations}\label{}
		\begin{align}
			&q_3'(L) \xrightarrow[]{s_4} q_4 \xrightarrow[]{s_5} q_4 \text{ in }\Acal^{\frakDp}_{\fsf},\\
			&q_5 \xrightarrow[]{s_6} q_5 \text{ in }\Acal^{\frakDp}_{\fsf}
		\end{align}
	\end{subequations}
	such that $q_5$ appears in $q_1'(L)\xrightarrow[]{s_2'(L)}q_2'(L)$,
	$s_4\in(E_{uo})^*$, $s_5,s_6\in(E_{uo})^+$, $\WEI(q_4 \xrightarrow[]{s_5} q_4)=:l_1$,
	$\WEI(q_5 \xrightarrow[]{s_6} q_5)=:l_2$, $l_1$ is eventually dominant, $l_2$ is eventually anti-dominant.  

	By \eqref{eqn22_diag_LMPautomata}, in $\Acal^{\frakDp}$ there exist paths
	\begin{align*}
		&q_0'(L)\xrightarrow[]{\bar s_1} q_1'(L)\xrightarrow[]{s_2'(L)_1} 
		q_5 \xrightarrow[]{s_6} q_5\xrightarrow[]{s_2'(L)_2}
		q_2'(L)
		\xrightarrow[]{e_f}q_3'(L)\\
		&\xrightarrow[]{s_4} q_4 \xrightarrow[]{s_5} q_4, \\
		&q_0'(R)\xrightarrow[]{\hat s_1} q_1'(R)\xrightarrow[]{s_2'(R)} q_2'(R)\xrightarrow[]{\ep}
		 q_3'(R)\\
	\end{align*}
	such that
	\begin{align*}
		&q_0'(L),q_0'(R)\in Q_0,\\
		&\bar s_1\in E^*E_o, \hat s_1\in(E\setminus E_f)^*(E_o\setminus E_f),\\
		&\ell(\bar s_1)=\ell(\hat s_1)=\ell(s_1'),\\
		&\ell(\tau(q_{0}'(L)\xrightarrow[]{\bar s_{1}} q_{1}'(L)))=
		\ell(\tau(q_{0}'(R)\xrightarrow[]{\hat s_{1}} q_{1}'(R))),\\
		&\WEI(q_{0}'(L)\xrightarrow[]{\bar s_{1}} q_{1}'(L))=
		\WEI(q_{0}'(R)\xrightarrow[]{\hat s_{1}} q_{1}'(R)),\\
		&q_2'(R)=q_3'(R),\quad (e_f,\ep)= e_3',\\
		&q_1'(L)\xrightarrow[]{s_2'(L)_1} q_5 \xrightarrow[]{s_2'(L)_2} q_2'(L) =
		q_1'(L)\xrightarrow[]{s_2'(L)} q_2'(L).
	\end{align*}

	Let $t\in T$ be eventually dominant. Choose sufficiently large $n\in\N$ such that
	$(l_1)^n\succ t$ and $\WEI(q_3'(L) \xrightarrow[]{s_4}q_4)\otimes (l_1)^n\succ {\bf1}$.
	Then choose sufficiently large $m\in\N$ such that
	$\WEI(q_1'(L)\xrightarrow[]{s_2'(L)_1} q_5 )\otimes (l_2)^m\prec{\bf1}$, $(l_2)^m\otimes\WEI(
	q_5\xrightarrow[]{s_2'(L)_2}q_2'(L) \xrightarrow[]{e_f}q_3'(L))\prec{\bf1}$, and $(l_2)^m\otimes
	\WEI\left(q_3'(L) \xrightarrow[]{s_4} q_4\left(  \xrightarrow[]{s_5} q_4 \right)^{2n}\right)\prec\\
	\WEI\left( q_1'(R)\xrightarrow[]{s_2'(R)} q_2'(R) \right)$.

	Denote
	\begin{align*}
		q_0'(L)\xrightarrow[]{\bar s_1} q_1'(L)\xrightarrow[]{s_2'(L)_1} 
		q_5 \left( \xrightarrow[]{s_6} q_5 \right) ^{3m}&\\
		\xrightarrow[]{s_2'(L)_2}
		q_2'(L)\xrightarrow[] {e_f}q_3'(L) &=: \pi_t,\\
		q_3'(L) \xrightarrow[]{s_4}q_4 \left(\xrightarrow[]
		{s_5}q_4 \right)^{2n}
		&=: \pi_t',\\
		q_0'(R)\xrightarrow[]{\hat s_1} q_1'(R)\xrightarrow[]{s_2'(R)} q_2'(R)
		&=: \pi_t''.
	\end{align*}
	One then has $E_f\notin\hat s_1s_2'(R)$, $\ell(\tau(\pi_t\pi_t'))=\ell(\tau(\pi_t''))=
	\ell(\tau(q_0'(L)\xrightarrow[]{\bar s_1} q_1'(L)))$,
	$\WEI_{\pi_t'} = \WEI(q_3'(L) \xrightarrow[]{s_4}q_4)\otimes (l_1)^{2n}
	\succeq (l_1)^n \succ t$, $\WEI(q_1'(L)\xrightarrow[]{s_2'(L)_1} q_5 \left( \xrightarrow[]{s_6} q_5 \right) ^{3m}
	\xrightarrow[]{s_2'(L)_2} q_2'(L)\xrightarrow[] {e_f}q_3'(L))\otimes\WEI_{\pi_t'} \preceq (l_2)^m\otimes 
	\WEI_{\pi_t'}\preceq \WEI\left( q_1'(R)\xrightarrow[]{s_2'(R)} q_2'(R) \right)$, $\WEI_{\pi_t''}=
	\WEI(q_{0}'(R)\xrightarrow[]{\hat s_{1}} q_{1}'(R))\otimes \WEI\left( q_1'(R)\xrightarrow[]{s_2'(R)} q_2'(R)
	\right)\succeq \WEI(q_{0}'(L)\xrightarrow[]{\bar s_{1}} q_{1}'(L))\otimes \WEI(q_1'(L)\xrightarrow[]{s_2'(L)_1} 
	q_5 \left( \xrightarrow[]{s_6} q_5 \right) ^{3m} \xrightarrow[]{s_2'(L)_2} q_2'(L)\xrightarrow[] {e_f}q_3'(L))
	\otimes\WEI_{\pi_t'} =\WEI_{\pi_t}\otimes\WEI_{\pi_t'} 
	\succeq \WEI_{\pi_t}\otimes t$.

	Also by Proposition~\ref{prop2_diag_LMPautomata}, $\Acal^{\frakDp}$ is not $E_f$-diagnosable.

	``only if'':

	Assume that $\Acal^{\frakDp}$ is not $E_f$-diagnosable. We choose an arbitrarily large eventually 
	dominant ${\bf1}\prec t\in T$.
	This can be done because for every eventually dominant element $t'$ in $T$ , one has
	$t'\prec (t')^2 \prec (t')^3 \prec \cdots$ (by Lemma~\ref{lem7_diag_LMPautomata})
	and every $(t')^n$ with $n>1$ is eventually dominant.
	By Proposition~\ref{prop2_diag_LMPautomata}, 
	there exist paths $\pi_t= q_0^1\xrightarrow[]{s}q_1^1$, $\pi_t'= q_1^1\xrightarrow[]{s'}q_2^1$, and
	$\pi_t''= q_0^2\xrightarrow[]{s''}q_2^2$, such that $q_0^1,q_0^2\in Q_0$, $\last(s)\in E_f$,
	$\ell(\tau(\pi_t''))=\ell(\tau(\pi_t\pi_t'))$,
	$E_f\notin s''$, $\WEI_{\pi_t'}\succ t$, and $\WEI_{\pi_t''}\succeq\WEI_{\pi_t}\otimes t$.

	By definition of $\CCa(\Acal^{\frakDp}_{\fsf},\Acal^{\frakDp}_{\nsf})$, the two paths $\pi_t\pi_t'$ and $\pi_t''$
	generate a run $\pi'$ of $\CCa(\Acal^{\frakDp}_{\fsf},\Acal^{\frakDp}_{\nsf})$ consisting of a sequence of 
	observable transitions followed by a sequence of unobservable transitions. Note that particularly, $\pi'$
	may contain only observable transitions or only unobservable transitions. Hence $\pi_t\pi_t'$ and $\pi_t''$
	can be rewritten as 
	\begin{subequations}\label{eqn18_diag_LMPautomata} 
		\begin{align}
			&\bar q_0\xrightarrow[]{\bar s_1} \cdots \xrightarrow[]{\bar s_n}\bar q_n \xrightarrow[]{\bar s_{n+1}}\bar q_{n+1},\\
			&\hat q_0\xrightarrow[]{\hat s_1} \cdots \xrightarrow[]{\hat s_n}\hat q_n \xrightarrow[]{\hat s_{n+1}}\hat q_{n+1},
		\end{align}
	\end{subequations}
	respectively, where $\bar q_0=q_0^1$, $\bar q_{n+1}=q_2^1$, $\hat q_0=q_0^2$, $\hat q_{n+1}=q_2^2$, 
	$n\in\N$, for all $i\in\llb 1,n\rrb$, $\bar s_i\in(E_{uo})^*E_o$,
	$\hat s_i\in(E_{uo}\setminus E_f)^*
	(E_o\setminus E_f)$, 
	$\WEI(\bar q_{i-1}\xrightarrow[]{\bar s_i}\bar q_i)=\WEI(\hat q_{i-1}\xrightarrow[]{\hat s_i}\hat q_i)$,
	$\bar s_{n+1}\in(E_{uo})^*$, $\hat s_{n+1}\in(E_{uo}\setminus E_f)^*$.
	Then 
	\begin{align}
		ss' &= \bar s_1 \dots \bar s_n \bar s_{n+1},\\
		s'' &= \hat s_1 \dots \hat s_n \hat s_{n+1},
	\end{align}
	and
	\begin{align*}
		\pi' =& (\bar q_0,\hat q_0)\xrightarrow[]{(\last(\bar s_1),\last(\hat s_1))} \cdots
		\xrightarrow[]{(\last(\bar s_n),\last(\hat s_n))}\\
		&(\bar q_n,\hat q_n) \xrightarrow[]{s_{n+1}'}(\bar q_{n+1},\hat q_{n+1}),
	\end{align*}
	where $s_{n+1}'\in(E_{uo}')^*$, $s_{n+1}'(L)=\bar s_{n+1}$, $s_{n+1}'(R)=\hat s_{n+1}$,
	for every $i\in\llb 1,n\rrb$, $(\bar q_{i-1},\hat q_{i-1})\xrightarrow[]{(\last(\bar s_i),\last(\hat s_i))}
	(\bar q_{i},\hat q_{i})$ is an observable transition, $(\bar q_n,\hat q_n) \xrightarrow[]{s_{n+1}'}
	(\bar q_{n+1},\hat q_{n+1})$ is a sequence of unobservable transitions.
	
	The subsequent argument is divided into four cases.

	\begin{enumerate}[(a)]
		\item\label{item21_diag_LMPautomata}
			Case $\bar s_1\dots \bar s_n= s$:

			In this case, $\last(\bar s_n)\in E_f$, $\pi_t=\bar q_0\xrightarrow[]{\bar s_1} \cdots 
			\xrightarrow[]{\bar s_n}\bar q_n$, $\pi_t'=\bar q_n\xrightarrow[]{\bar s_{n+1}}\bar q_{n+1}$,
			$\WEI_{\pi_t}=\WEI(\hat q_0\xrightarrow[]{\hat s_1} \cdots \xrightarrow[]{\hat s_n}\hat q_n)$,
			$\WEI_{\pi_t'}\succ t$,
			$\WEI_{\pi_t''}=\WEI_{\pi_t}\otimes \WEI(\hat q_n \xrightarrow[]{\hat s_{n+1}}\hat q_{n+1})\succeq
			\WEI_{\pi_t}\otimes t$. Then we have $\WEI(\hat q_n \xrightarrow[]{\hat s_{n+1}}\hat q_{n+1})
			\succeq t$ because a path always has nonzero weight (because $\frakDp$ has no zero divisor)
			and $\frakDp$ is cancellative (Lemma~\ref{lem6_diag_LMPautomata}). Recall that $t$ is
			arbitrarily large. Divide $\hat q_n \xrightarrow[]{\hat s_{n+1}}\hat q_{n+1}$ to paths 
			$\hat\pi_1\dots\hat\pi_m$ such that $\last(\hat\pi_i)=\init(\hat\pi_{i+1})$ for all $i\in\llb 1,m-1
			\rrb$, where $\hat\pi_1$ is the shortest prefix of $\hat q_n \xrightarrow[]{\hat s_{n+1}}\hat q_{n+1}$ 
			having weight being eventually dominant, $\hat\pi_2$ is the shortest prefix of 
			$(\hat q_n \xrightarrow[]{\hat s_{n+1}}\hat q_{n+1})\setminus \pi_1$ having weight being
			eventually dominant, \dots, $\hat\pi_{m-1}$ is the shortest prefix of  
			$(\hat q_n \xrightarrow[]{\hat s_{n+1}}\hat q_{n+1})\setminus (\pi_1\dots\pi_{m-2})$ having weight
			being eventually dominant. Then we have $\WEI_{\hat\pi_1}, \dots,
			\WEI_{\hat\pi_{m-1}}$ are eventually dominant,
			but $\WEI_{\hat\pi_m}$ is not necessarily. We also have $m$ can be arbitrarily large,
			because $t$ is arbitrarily large and the product of any finitely many elements that are not 
			eventually dominant cannot be very large (as shown in Remark~\ref{rem1_diag_LMPautomata}). 
			Then there exist repetitive states among $\init(\hat\pi_1)$,
			$\init(\hat\pi_2)$, \dots, $\init(\hat\pi_{m-1})$, and all these repetitive states are eventually
			dominant
			by Lemma~\ref{lem1_diag_LMPautomata}. That is, $\init(\hat\pi_1)=\hat q_n$ is eventually dominant
			in $\Acal^{\frakDp}_{\nsf}$. Similarly, we have $\init(\pi_t')=\bar q_n$ is eventually dominant
			in $\Acal^{\frakDp}_{\fsf}$. Then \eqref{item7_diag_LMPautomata} holds, and the $s_3'$ in 
			\eqref{eqn8_diag_LMPautomata} is equal to $\ep$.

		\item\label{item22_diag_LMPautomata}
			Case $\bar s_1\dots \bar s_n\sqsubsetneq s$:

			In this case, denote $\pi_t:=\bar q_0\xrightarrow[]{\bar s_1} \cdots 
			\xrightarrow[]{\bar s_n}\bar q_n\pi_1'$, then $\pi_t'=(\bar q_n\xrightarrow[]{\bar s_{n+1}}\bar q_{n+1})
			\setminus \pi_1'$. We have $\WEI_{\pi_t'}\succ t$, $\WEI(\bar q_0\xrightarrow[]{\bar s_1} \cdots 
			\xrightarrow[]{\bar s_n}\bar q_n) = \WEI(\hat q_0\xrightarrow[]{\hat s_1} \cdots \xrightarrow[]{\hat s_n}
			\hat q_n)$, $\WEI_{\pi_t''}=\WEI(\hat q_0\xrightarrow[]{\hat s_1} \cdots \xrightarrow[]{\hat s_n}
			\hat q_n)\otimes \WEI(\hat q_n \xrightarrow[]{\hat s_{n+1}}\hat q_{n+1}) \succeq \WEI_{\pi_t}\otimes t
			= \WEI(\bar q_0\xrightarrow[]{\bar s_1} \cdots \xrightarrow[]{\bar s_n}\bar q_n)\otimes \WEI_{\pi_1'}
			\otimes t$. Using the argument as in \eqref{item21_diag_LMPautomata}, we know $\last(\pi_1')$ is eventually
			dominant in $\Acal^{\frakDp}_{\fsf}$. By Lemma~\ref{lem6_diag_LMPautomata}, we have $\WEI(\hat q_n
			\xrightarrow[]{\hat s_{n+1}}\hat q_{n+1}) \succeq \WEI_{\pi_1'} \otimes t$.

			\begin{itemize}
				\item
			Assume $\hat q_n \xrightarrow[]{\hat s_{n+1}}\hat q_{n+1}$ has a prefix $\pi_2'$ such that $\WEI_{\pi_2'}
			\preceq\WEI_{\pi_1'}$. Then we have $\WEI(\hat q_n	\xrightarrow[]{\hat s_{n+1}}\hat q_{n+1}) =
			\WEI_{\pi_2'} \otimes \WEI((\hat q_n \xrightarrow[]{\hat s_{n+1}}\hat q_{n+1})\setminus\pi_2')
			\succeq \WEI_{\pi_1'}\otimes t$. By $\WEI_{\pi_2'} \preceq \WEI_{\pi_1'}$ and
			Lemma~\ref{lem6_diag_LMPautomata}, we have $\WEI((\hat q_n \xrightarrow[]{\hat s_{n+1}}
			\hat q_{n+1})\setminus \pi_2') \succeq t$. Also by using the argument as in \eqref{item21_diag_LMPautomata},
			$\last(\pi_2')$ is eventually dominant in $\Acal^{\frakDp}_{\nsf}$. That is, 
			\eqref{item8_diag_LMPautomata} holds.

				\item
			Assume $\hat q_n \xrightarrow[]{\hat s_{n+1}}\hat q_{n+1}$ has no prefix with weight no greater than 
			$\WEI_{\pi_1'}$. Then $\WEI_{\pi_1'}\prec{\bf1}$. If 
			$\hat q_n \xrightarrow[]{\hat s_{n+1}}\hat q_{n+1}$ contains a dominant cycle in $\Acal_{\nsf}^{\frakDp}$,
			then \eqref{item8_diag_LMPautomata} holds; otherwise, we must have $\WEI_{\pi_1'}$ is sufficiently small, 
			then also using the argument as in \eqref{item21_diag_LMPautomata} we have $\pi_1'$ contains
			an anti-dominant cycle in $\Acal_{\fsf}^{\frakDp}$, i.e., \eqref{item27_diag_LMPautomata} holds.
			\end{itemize}

		\item\label{item23_diag_LMPautomata}
			Case $s\sqsubsetneq \bar s_1\dots \bar s_n$ and $\bar s_{n+1}=\ep$:
			
			Choose $m\in\llb 1,n\rrb$ such that $\bar s_1\dots \bar s_{m-1}\sqsubset
			s\sqsubsetneq \bar s_1\dots \bar s_m$.

			\begin{itemize}
				\item Assume $\bar s_1\dots \bar s_{m-1}=s$. 

					If there exists $l\in\llb m,n\rrb$ such that $\WEI(\bar q_{l-1}\xrightarrow[]{\bar s_l} \bar q
					_l)=\WEI(\hat q_{l-1}\xrightarrow[]{\hat s_l} \hat q_l)$ is sufficiently large, then $\bar q_
					{l-1}\xrightarrow[]{\bar s_l} \bar q_l$ contains a dominant cycle in $\Acal_{\fsf}^{\frakDp}$,
					$\hat q_{l-1}\xrightarrow[]{\hat s_l} \hat q_l$ contains a dominant cycle in 
					$\Acal_{\nsf}^{\frakDp}$, then 
					\eqref{item7_diag_LMPautomata} holds. Otherwise if such $l$ does not exist, then 	
        			the weights of $\bar q_{m-1}\xrightarrow[] {\bar s_m}\bar q_m$, $\dots$, $\bar q_{n-1}
        			\xrightarrow[] {\bar s_n}\bar q_n$ cannot be arbitrarily large. Because
        			$t$ is arbitrarily large and $t\prec\WEI(\bar q_{m-1}\xrightarrow[] {\bar s_m}\cdots\xrightarrow[] 
					{\bar s_n}\bar q_n)$, by using the argument as in \eqref{item21_diag_LMPautomata},
					in $(\bar q_{m-1},\hat q_{m-1})\xrightarrow[]{(\last(\bar s_m),
        			\last(\hat s_m))} \cdots \xrightarrow[]{(\last(\bar s_n),\last(\hat s_n))}(\bar q_n,\hat q_n)$
					there exists a dominant cycle consisting of observable transitions.
					Hence \eqref{item6_diag_LMPautomata} holds.
        		\item Assume $\bar s_1\dots \bar s_{m-1}\sqsubsetneq s$.
        			From $\bar q_{m-1}\xrightarrow[]{\bar s_m} \bar q_m$ and $\hat q_{m-1}\xrightarrow[]{\hat s_m} \hat q_m$,
        			by Proposition~\ref{prop1_diag_LMPautomata},
        			we obtain an unobservable sequence $(\bar q_{m-1},\hat q_{m-1})\xrightarrow[]{s_m^1}\bar v
					\xrightarrow[]{s_m^2}(\bar q_m',\hat q_m')=:\bar \pi$, where
        			$s_m^1(L)$ ($\ne\ep$) is a suffix of $s$, $s_m^1(R)=\ep$, and
        			\begin{align*}
						&\bar \pi(L)\xrightarrow[]{\last(\bar s_m)}\bar q_m = \bar q_{m-1}\xrightarrow[]{\bar s_m}\bar q_m,\\
        				&\bar \pi(R)\xrightarrow[]{\last(\hat s_m)}\hat q_m=\hat q_{m-1}\xrightarrow[]{\hat s_m}\hat q_m.\\
						&\text{(the path before $=$ is the same as the one after $=$)}
        			\end{align*}
					If there exist dominant cycles in both components of $\bar v\xrightarrow[]{s_m^2}(\bar 
					q_m',\hat q_m')$, \eqref{item8_diag_LMPautomata} holds; if there exists a dominant cycle
					in $\bar v(L)\xrightarrow[]{s_m^2(L)} \bar q_m'$ and there exists an anti-dominant cycle
					in $\bar q_{m-1}\xrightarrow[]{s_m^1(L)}\bar v(L)$, \eqref{item27_diag_LMPautomata} holds;
					otherwise, 
					also by argument similar as above, either \eqref{item6_diag_LMPautomata} or 
        			\eqref{item7_diag_LMPautomata} holds.
        	\end{itemize}

		\item\label{item24_diag_LMPautomata}
			Case $s\sqsubsetneq \bar s_1\dots \bar s_n$ and $\bar s_{n+1}\ne\ep$:

			Using argument similar as above, one also has at least one of \eqref{item6_diag_LMPautomata},
			\eqref{item7_diag_LMPautomata}, \eqref{item8_diag_LMPautomata}, and \eqref{item27_diag_LMPautomata}
			holds. This finishes the proof.
		\end{enumerate}

\end{document}